\documentclass{article}[11pt]
\usepackage{fullpage}

\usepackage{amsthm}
\usepackage{amssymb}
\usepackage{amsmath}
\usepackage{mathtools}
\usepackage{hyperref}
\usepackage{xcolor}
\usepackage{xspace}

\newcommand{\defn}[1]{{\textit{\textbf{\boldmath #1}}}\xspace}
\renewcommand{\paragraph}[1]{\vspace{0.09in}\noindent{\bf \boldmath #1.}}
\renewcommand{\epsilon}{\varepsilon}

\DeclareMathOperator{\polylog}{\textnormal{polylog}}
\DeclareMathOperator{\poly}{\textnormal{poly}}

\newcommand{\floor}[1]{\left\lfloor #1 \right\rfloor}
\newcommand{\ceil}[1]{\left\lceil #1 \right\rceil}
\newcommand{\paren}[1]{\left( #1 \right)}

\DeclareMathOperator{\fil}{\text{\textup{fill}}}
\DeclareMathOperator{\rep}{\text{\textup{rep}}}
\DeclareMathOperator{\alg}{\text{\textup{alg}}}
\DeclareMathOperator{\randalg}{\text{\textup{randalg}}}
\DeclareMathOperator{\flatalg}{\text{\textup{flatalg}}}
\DeclareMathOperator{\trivalg}{\text{\textup{trivalg}}}

\DeclareMathOperator{\skips}{\text{\textup{skips}}}

\usepackage[capitalise,nameinlink,noabbrev]{cleveref}
\crefname{equation}{}{} 
\crefname{enumi}{Step}{} 

\usepackage{algorithm}
\usepackage[noend]{algpseudocode} 

\newtheorem{definition}{Definition}[section]

\newtheorem{proposition}{Proposition}[section]
\newtheorem{clm}{Claim}[section]
\newtheorem{lemma}{Lemma}[section]
\newtheorem{corollary}{Corollary}[section]
\newtheorem{theorem}{Theorem}[section]

\usepackage{csquotes}
\usepackage{authblk}

\usepackage{fancyhdr}
\usepackage[hang,flushmargin]{footmisc} 

\title{The Variable-Processor Cup Game}
\date{\vspace{-5ex}}

\author{William Kuszmaul\thanks{Supported by a Hertz fellowship
and a NSF GRFP fellowship. Research was sponsored in part by the United States Air Force Research Laboratory and the United States Air Force Artificial Intelligence Accelerator and was accomplished under Cooperative Agreement Number FA8750-19-2-1000. The views and conclusions contained in this document are those of the authors and should not be interpreted as representing the official policies, either expressed or implied, of the United States Air Force or the U.S. Government. The U.S. Government is authorized to reproduce and distribute reprints for Government purposes notwithstanding any copyright notation herein.}, Alek Westover\thanks{Partially
supported by MIT PRIMES.}}

\affil{Massachusetts Institute of Technology}
\affil{\small \texttt{kuszmaul@mit.edu, alek.westover@gmail.com}}

\begin{document}
\maketitle
\begin{abstract}
  
The problem of scheduling tasks on $p$ processors so that no
task ever gets too far behind is often described as a game
with cups and water. In the $p$-processor cup game on $n$ cups,
there are two players, a filler and an emptier, that take turns
adding and removing water from a set of $n$ cups. In each turn,
the filler adds $p$ units of water to the cups, placing at most
$1$ unit of water in each cup, and then the emptier selects $p$
cups to remove up to $1$ unit of water from. The emptier's goal
is to minimize the backlog, which is the height of the fullest
cup.

The $p$-processor cup game has been studied in many different
settings, dating back to the late 1960's. All of the past work
shares one common assumption: that $p$ is fixed. This paper
initiates the study of what happens when the number of available
processors $p$ varies over time, resulting in what we call the
\emph{variable-processor cup game}.

Remarkably, the optimal bounds for the variable-processor cup
game differ dramatically from its classical counterpart. Whereas
the $p$-processor cup has optimal backlog $\Theta(\log n)$, the
variable-processor game has optimal backlog $\Theta(n)$.
Moreover, there is an efficient filling strategy that yields
backlog $\Omega(n^{1 - \epsilon})$ in quasi-polynomial time against any
deterministic emptying strategy.

We additionally show that straightforward uses of randomization
cannot be used to help the emptier. In particular, for any
positive constant $\Delta$, and any $\Delta$-greedy-like randomized
emptying algorithm $\mathcal{A}$, there is a filling strategy
that achieves backlog $\Omega(n^{1 - \epsilon})$ against $\mathcal{A}$ in
quasi-polynomial time.
\end{abstract}

\thispagestyle{fancy} 

\section{Introduction}\label{sec:intro}

A fundamental challenge in processor scheduling is how to perform
load balancing, that is, how to share processors among tasks in
order to keep any one task from getting too far behind. Consider
$n$ tasks executing in time slices on $p < n$ processors. During
each time slice, a scheduler must select $p$ (distinct) tasks
that will be executed during the slice; up to one unit of work is
completed on each executed task.
During the same time slice, however, up to $p$ units of \emph{new
work} may be allocated to the tasks, with different tasks
receiving different amounts of work. The goal of a load-balancing
scheduler is to bound the backlog of the system, which is defined
to be the maximum amount of uncompleted work for any task.

As a convention, the load balancing problem is often described as
a game involving water and cups. The \defn{$p$-processor cup
game} is a multi-round game with two players, an \defn{emptier}
and a \defn{filler}, that takes place on $n$ initially empty
cups. At the beginning of each round, the filler adds up to $p$
units of water to the cups, subject to the constraint that each
cup receives at most $1$ unit of water. The emptier then selects
up to $p$ distinct cups and removes up to $1$ unit of water from
each of them. The emptier's goal is to minimize the amount of
water in the fullest cup, also known as the \defn{backlog}. In
terms of processor scheduling, the cups represent tasks, the
water represents work assigned to each task, and the emptier
represents a scheduling algorithm.

Starting with the seminal paper of Liu \cite{Liu69}, work on the $p$ processor cup game has spanned more than five decades \cite{BaruahCoPl96,GkasieniecKl17,BaruahGe95,LitmanMo11,LitmanMo05,MoirRa99,BarNi02,GuanYi12,Liu69, LiuLa73,DietzRa91, BenderFaKu19, Kuszmaul20, AdlerBeFr03, DietzSl87, LitmanMo09}. In addition to processor scheduling \cite{BaruahCoPl96,GkasieniecKl17,BaruahGe95,LitmanMo11,LitmanMo05,MoirRa99,BarNi02,GuanYi12,Liu69, LiuLa73, AdlerBeFr03, LitmanMo09, DietzRa91}, applications include network-switch buffer management~\cite{Goldwasser10,AzarLi06,RosenblumGoTa04,Gail93}, quality of service guarantees~\cite{BaruahCoPl96,AdlerBeFr03,LitmanMo09}, and data structure deamortization \cite{AmirFaId95,DietzRa91,DietzSl87,AmirFr14,Mortensen03,GoodrichPa13,FischerGa15,Kopelowitz12,BenderDaFa20}.

The game has also been studied in many different forms. Researchers have studied the game with a fixed-filling-rate constraint \cite{BaruahCoPl96,GkasieniecKl17,BaruahGe95,LitmanMo11,LitmanMo05,MoirRa99,BarNi02,GuanYi12,Liu69, LiuLa73}, with various forms of resource augmentation \cite{BenderFaKu19, Kuszmaul20, LitmanMo09, DietzRa91}, with both oblivious and adaptive adversaries \cite{AdlerBeFr03,BaruahCoPl96,Liu69,BenderFaKu19, Kuszmaul20}, with smoothed analysis \cite{Kuszmaul20, BenderFaKu19}, with a semi-clairvoyant emptier \cite{LitmanMo09}, with competitive analysis \cite{Bar-NoyFrLa02, FleischerKo04, DamaschkeZh05}, etc.

For the plain form of the $p$-processor cup game, the greedy
emptying algorithm (i.e., always empty from the fullest cups) is
known to be asymptotically optimal \cite{AdlerBeFr03,
BenderFaKu19, Kuszmaul20}, achieving backlog $O(\log n)$. The
optimal backlog for randomized emptying algorithms remains an
open question \cite{DietzRa91, BenderFaKu19, Kuszmaul20} and is
known to be between $\Omega(\log \log n)$ and $O(\log \log n +
\log p)$ \cite{Kuszmaul20}.

\paragraph{This paper: varying resources}
Although cup games have been studied in many forms, all of the
prior work on cup games shares one common assumption: the number
$p$ of processors is fixed.

In modern computing, however, computers are often shared among
multiple applications, users, and even virtual OS's. The result
is that the amount of resources (e.g., memory, processors, cache,
etc.) available to a given application may fluctuate over time.
The problem of handling cache fluctuations has received extensive
attention in recent years (see work on cache-adaptive analysis
\cite{CA1, CA2, CA3, CA4, CA5}), but the problem of handling a
varying number of processors remains largely unstudied.

This paper introduces the \defn{variable-processor cup game}, in
which the filler is allowed to \emph{change} $p$ (the total
amount of water that the filler adds, and the emptier removes,
from the cups per round) at the beginning of each round. Note
that we do not allow the resources of the filler and emptier to
vary separately. That is, as in the standard game, we take the
value of $p$ for the filler and emptier to be identical in each
round. This restriction is crucial since, if the filler is
allowed more resources than the emptier, then the filler could
trivially achieve unbounded backlog.

A priori having variable resources offers neither player a clear
advantage. When the number $p$ of processors is fixed, the greedy
emptying algorithm (i.e., always empty from the fullest cups), is
known to achieve backlog $O(\log n)$ \cite{AdlerBeFr03,
BenderFaKu19, Kuszmaul20} regardless of the value of $p$. This
seems to suggest that, when $p$ varies, the correct backlog
should remain $O(\log n)$. In fact, when we began this project,
we hoped for a straightforward reduction between the two versions
of the game.

\paragraph{Results}
We show that the variable-processor cup game is \emph{not}
equivalent to the standard $p$-processor game. By strategically
controlling the number $p$ of processors, the filler can achieve
substantially larger backlog than would otherwise be possible.

We begin by constructing filling strategies against deterministic
emptying algorithms. We show that for any positive constant
$\epsilon$, there is a filling strategy that achieves backlog
$\Omega(n^{1 - \epsilon})$ within $2^{\polylog(n)}$ rounds.
Moreover, if we allow for $n!$ rounds, then there is a filling
strategy that achieves backlog $\Omega(n)$. In contrast, for the
$p$-processor cup game with any fixed $p$, the backlog never
exceeds $O(\log n)$.

Our lower-bound construction is asymptotically optimal. By
introducing a novel set of invariants, we deduce that the greedy
emptying algorithm never lets backlog exceed $O(n)$.

A natural question is whether \emph{randomized} emptying
algorithms can do better. In particular, when the emptier is
randomized, the filler is taken to be \defn{oblivious}, meaning
that the filler cannot see what the emptier does at each step.
Thus the emptier can potentially use randomization to obfuscate
their behavior from the filler, preventing the filler from being
able to predict the heights of cups.

When studying randomized emptying strategies, we restrict
ourselves to the class of \defn{greedy-like emptying strategies}.
This means that the emptier never chooses to empty from a cup $c$
over another cup $c'$ whose fill is more than $\omega(1)$ greater
than the fill of $c$. All of the known randomized emptying
strategies for the classic $p$-processor cup game are greedy-like
\cite{BenderFaKu19, Kuszmaul20}.

Remarkably, the power of randomization does not help the emptier
very much in the variable-processor cup game. For any constant
$\epsilon > 0$, we give an oblivious filling strategy that
achieves backlog $\Omega(n^{1 - \epsilon})$ in quasi-polynomial
time (with probability $1 - 2^{-\polylog n}$), no matter what
(possibly randomized) greedy-like strategy the emptier follows. 

Our results combine to tell a surprising story. They suggest that
the problem of varying resources poses a serious theoretical
challenge for the design and analysis of load-balancing
scheduling algorithms. There are many possible avenues for future
work. Can techniques from beyond worst-case analysis (e.g.,
smoothing, resource augmentation, etc.) be used to achieve better
bounds on backlog? What about placing restrictions on the filler
(e.g., bounds on how fast $p$ can change), allowing the emptier
to be able to be semi-clairvoyant, or making stochastic
assumptions on the filler? We believe that all of these questions
warrant further attention.

\section{Preliminaries}\label{sec:prelims}
The cup game consists of a sequence of rounds. Let $S_t$ denote the state of of the cups at the start of round $t$. At the beginning of the round, the filler chooses the number
of processors $p_t$ for the round. Next, the filler distributes
$p_t$ units of water among the cups (with at most $1$ unit of
water to any particular cup). Now the game is at the
\defn{intermediate state} for round $t$, which we call state $I_t$.
Finally the emptier chooses $p_t$ cups to empty at most $1$ unit
of water from, which marks the conclusion of round $t$. The state
is then $S_{t+1}$.

If the emptier empties from a cup $c$ on round $t$ such that the fill
of $c$ is less than $1$ in state $I_t$, then $c$ now has $0$ fill (not
negative fill); we say that the emptier \defn{zeroes out} $c$ on round
$t$. Note that on any round where the emptier zeroes out a cup the
emptier has removed less total fill from the cups than the filler has
added to the cups; hence the average fill of the cups has increased.

We denote the fill of a cup $c$ at state $S$ by $\fil_S(c)$. Let
the \defn{mass} of a set of cups $X$ at state $S$ be $m_S(X) =
\sum_{c\in X} \fil_S(c)$. Denote the average fill of a set of
cups $X$ at state $S$ by $\mu_S(X)$. Note that $\mu_S(X) |X| =
m_S(X)$. Let the \defn{backlog} at state $S$ be $\max_c
\fil_S(c)$, let the \defn{anti-backlog} at state $S$ be $\min_c
\fil_S(c).$
Let the \defn{rank} of a cup at a given state be its position in
the list of the cups sorted by fill at the given state, breaking
ties arbitrarily but consistently. For example, the fullest cup
at a state has rank $1$, and the least full cup has rank $n$. Let
$[n] = \{1,2,\ldots, n\}$, let $i+[n] = \{i+1, i+2, \ldots,
i+n\}$. For a state $S$, let $S(r)$ denote the rank $r$ cup at
state $S$, and let $S(\{r_1,r_2,\ldots, r_m\})$ denote the set of
cups of ranks $r_1, r_2,\ldots r_m$.

As a tool in the analysis we define a new variant of the cup
game: the \defn{negative-fill} cup game. In the negative-fill cup
game, when the emptier empties from a cup, the cup's fill always
decreases by exactly $1$, i.e. there is no zeroing out. We refer
to the standard version of the cup game where cups can zero out
as the \defn{standard-fill} cup game when necessary for clarity.

The notion of negative fill will be useful in our lower-bound
constructions, in which we want to construct a strategy for the filler
that achieves large backlog. By analyzing a filling strategy on the
negative-fill game, we can then reason about what happens if we apply
the same filling strategy recursively to a set of cups $S$ whose
average fill $\mu$ is larger than $0$; in the recursive application,
the average fill $\mu$ acts as the ``new $0$'', and fills less than
$\mu$ act as negative fills.

Note that it is strictly easier for the filler to achieve high backlog
in the standard-fill cup game than in the negative-fill cup game;
hence a lower bound on backlog in the negative-fill cup game also
serves as a lower bound on backlog in the standard-fill cup game.  On
the other hand, during the upper bound proof we use the standard-fill
cup game: this makes it harder for the emptier to guarantee its upper
bound.

\paragraph{Other Conventions} When discussing the state of the cups
\defn{at a round \boldmath $t$}, we will take it as given that we are
referring to the starting state $S_t$ of the round. Also, when
discussing sets, we will use $XY$ as a shorthand for $X \cup
Y$. Finally, when discussing the average fill $\mu_{S_t}(X)$ of a set
of cups, we will sometimes ommit the subscript $S_t$ when the round
number is clear.

\section{Technical Overview}
\label{sec:technical_overview}

In this section we present proof sketches and discuss the main technical ideas for our results. 

\subsection{Adaptive Lower Bound}\label{sec:technicaladaptive}
In \cref{sec:adaptive} we provide an adaptive filling strategy that achieves
backlog $\Omega(n^{1 - \epsilon})$; in this subsection we sketch the
proof of the result.

First we note that there is a trivial algorithm, that we call
\defn{$\trivalg$}, for achieving backlog at least $1/2$ on at
least $2$ cups in time $O(1)$.

The essential ingredient in our lower-bound construction is what we
call the \defn{Amplification Lemma}. The lemma takes as input a
filling strategy that achieves some backlog curve $f$ (i.e., on $n$
cups, the strategy achieves backlog $f(n)$), and outputs a new filling
strategy that achieves a new amplified backlog curve $f'$.

\begin{lemma}[\cref{lem:adaptiveAmplification}]
  Let $\alg(f)$ be a filling strategy that achieves backlog
  $f(n)$ on $n$ cups (in the negative-fill cup game). There
  exists a filling strategy $\alg(f')$,
  the \defn{amplification} of $\alg(f)$, that achieves backlog at
  least $$f'(n) \ge (1-\delta) f(\floor{(1-\delta)n}) +
  f(\ceil{\delta n}).$$
\end{lemma}
\begin{proof}[Proof Sketch]
The filler designates an \defn{anchor set} $A$ of size
$\ceil{\delta n}$ and a \defn{non-anchor set} $B$ of size
$\floor{(1-\delta)n}$.

The filler's strategy begins with $M$ phases, for some parameter $M$
to be determined later. In each phase,
the filler applies $\alg(f)$ to the non-anchor set $B$, while
simultaneously placing $1$ unit of water into each cup of $A$ on each
step. If there is ever a step during the phase in which the emptier
does not remove water from every cup in $A$, then the phase is said to
be \defn{emptier neglected}. On the other hand, if a phase is not
emptier neglected, then at the end of the phase, the filler swaps the
cup in $B$ with largest fill with the cup in $A$ whose fill is
smallest.

After the $M$ phases are complete, the filler then recursively
applies $\alg(f)$ to the cups $A$. This completes the filling strategy
$\alg(f')$.

The key to analyzing $\alg(f')$ is to show that, at the end of the
$M$-th phase, the average fill of the cups $A$ satisfies
$\mu(A) \ge (1 - \delta) f(|B|)$. This, in turn, means that the
recursive application of $\alg(f)$ to $A$ will achieve backlog
$(1 - \delta) f(|B|) + f(|A|)$, as desired.

Now let us reason about $\mu(A)$. If a phase is emptier neglected,
then the total amount of water placed into $A$ during the phase is at
least $1$ greater than the total amount of water removed. Hence
$\mu(A)$ increases by at least $1/|A|$. On the other hand, if a phase
is not emptier neglected, then $\alg(f)$ will successfully achieve
backlog $\mu(B) + f(|B|)$ on the cups $B$ during the phase. At the end
of the phase, the filler will then swap a cup from $B$ with large fill
with a cup from $A$. Thus, in each phase, we either have that $\mu(A)$
increases by $1 / |A|$, or that a cup with large fill gets swapped
into $A$. After sufficiently many phases, one can show that $\mu(A)$
is guaranteed to become at least $(1 - \delta)f(|B|) + f(|A|)$.


  
\end{proof}

We use the Amplification Lemma to give two lower bounds on
backlog: one with reasonable running time, the other with
slightly better backlog.

\begin{theorem}[\cref{thm:adaptivePoly}]
  There is an adaptive filling strategy for achieving backlog
  $\Omega(n^{1-\varepsilon})$ for constant $\varepsilon \in (0,
  1/2)$ in running time $2^{O(\log^2 n)}$. 
\end{theorem}
\begin{proof}[Proof Sketch]
  We construct a sequence of filling strategies with $\alg(f_{i+1})$
  the amplification of $\alg(f_i)$ using $\delta = \Theta(1)$
  determined as a function of $\varepsilon$, and $\alg(f_0)=\trivalg$.
  Choosing $\delta$ appropriately as a function of $\epsilon$, and
  letting $c$ be some (small) positive constant, we show by induction
  on $i$ that, for all $k \le 2^{c i}$, $\alg(f_{i})$ achieves backlog
  $\Omega(k^{1 - \epsilon})$ on $k$ cups in running time
  $2^{O(\log^2 k)}$. Taking $i = \Theta(\log n)$ completes the proof.
\end{proof}

\begin{theorem}[\cref{thm:factorialTimeAlg}]
  There is an adaptive filling strategy for achieving backlog
  $\Omega(n)$ in running time $O(n!)$.
\end{theorem}
\begin{proof}[Proof Sketch]
We construct a sequence of filling strategies with $\alg(f_{i+1})$ the
amplification of $\alg(f_i)$ using $\delta = 1/(i+1)$, and
$\alg(f_0)$ a filling strategy for achieving backlog $1$ on
$O(1)$ cups in $O(1)$ time (this is a slight modification of
$\trivalg$). We show by induction that
$\alg(f_{\Theta(n)})$ achieves backlog $\Omega(n)$ in running
time $O(n!)$.
\end{proof}

\subsection{Upper Bound}
In \cref{sec:upperBound} we prove that a greedy emptier, i.e. an
emptier that always empties from the $p$ fullest cups, never lets
backlog exceed $O(n)$; in this subsection we sketch the proof of
this result.

\cref{thm:TO_invariant} gives a system of invariants on the state of the
cups after each step. By considering the invariant at $k = 1$, we
achieve a bound of $O(n)$ on the backlog.

\begin{theorem}[\cref{thm:invariant}]
  \label{thm:TO_invariant}
 For all $k \le n$, the average fill of the $k$ fullest cups never exceeds $2n-k$ at the beginning of any round.
\end{theorem}
\begin{proof}[Proof Sketch]
 The proof is by induction on the round. Fix some
round $t$ and assume that the result holds for all $k$ at the beginning of round $t$. Fix
some $k$; we aim to prove that the average fill of the $k$
fullest cups is at most $2n-k$ at the start of round $t+1$. 

Let $A$ be the cups that satisfy the following three properties: they
are among the $k$ fullest cups in $I_t$, they are emptied from in step
$t$, and they are among the $k$ fullest cups in $S_{t+1}$. Let $B$ be
the cups that satisfy the following three properties: they are among
the $k$ fullest cups in state $I_t$, they are emptied from in step
$t$, and are \emph{not} among the $k$ fullest cups in
$S_{t+1}$. Finally, let $C$ be the cups with ranks
$|A| + |B| + 1, \ldots, k + |B|$ in state $I_t$. The set $C$ is
defined so that the $k$ fullest cups in state $S_{t+1}$ are given by $AC$.

For simplicity, throughout the rest of the proof we will make the
following assumption: the rank $r$ cup at state $S_t$ is also the rank
$r$ cup at state $I_t$ for all ranks $r \in [n]$. In the full version
of the proof, we show that this assumption is
actually without loss of generality.

We break the rest of the proof into three cases. Let $a = |A|$,
$b = |B|$, and $c = |C|$. 

\vspace{.3 cm}

\noindent\textbf{Case 1}:
Some cup in $A$ zeroes out in round $t$.\\
\textbf{Analysis}: The fill of all cups in $C$ must be at most $1$ at
state $I_t$ to be less than the fill of the cup in $A$ that zeroed
out. Furthermore, the average fill of $A$ at the beginning of step
$t + 1$ must be at most the average fill of the $a - 1$ fullest cups
in $S_t$ (due to one of the cups being zeroed out), and thus is at
most $2n-(a - 1)$. Combined with some algebra, these facts imply
that the average fill in $AC$ in  not too large, in particular not
larger than $2n-k$.

\vspace{.3 cm}

\noindent\textbf{Case 2}:
We have $b = 0$ and no cups in $A$ zero out in round $t$.\\
\textbf{Analysis}: In this case the set of cups of ranks in $[k]$ at
state $S_t$ is the same as the set of cups of ranks in $[k]$ at state
$S_{t+1}$, and these sets are both given by $AC$. During round $t$ the
emptier removes $a$ units of fill from the cups $A$. The filler cannot
have added more than $k$ fill to the cups $AC$, because it can add at
most $1$ fill to any given cup. Also, the filler cannot have added
more than $p_t$ fill to the cups because this is the total amount of
fill that the filler is allowed to add. Hence the filler adds at most
$\min(p_t, k) = a+b=a+0=a$ fill to the cups $AC$. It follows that the
emptier removes at least as many units of water from the cups $AC$ as
the filler adds, so the average fill of $AC$ has not increased and is
still at most $2n-k$.

\vspace{.3 cm}

\noindent\textbf{Case 3}:
We have $b > 0$ and no cups in $A$ zero out in round $t$. \\
\textbf{Analysis}:
This is the most interesting of the three cases. Consider
$m_{S_{t+1}}(AC)$, which is the mass of the $k$ fullest
cups after step $t$. Each cup in $A$ was emptied from in step $t$, and the
filler adds at most $\min(k, p_t) = a+b$ fill to cups $AC$ during the step. 
Hence, 
\begin{equation}
  \label{eq:TO_bplusmAC}
m_{S_{t+1}}(AC) \le m_{S_t}(AC) + b.
\end{equation}

Using the fact that $\mu_{S_t}(B) \ge \mu_{S_t}(C)$, we have that, 
$$m_{S_t}(C) \le \frac{c}{b+c} m_{S_t}(BC) = \frac{c}{b+c}(m_{S_t}(ABC) - m_{S_t}(A)).$$
Thus
\begin{equation}
  m_{S_{t}}(AC) = m_{S_t}(A) + m_{S_t}(C) \le \frac{c}{b+c}m_{S_t}(ABC) + \frac{b}{b+c}m_{S_t}(A).
  \label{eq:ACbound}
\end{equation} 
The system of invariants at at the beginning of step $t$ lets us
bound $m_{S_t}(A)$ by $|A|(2n - |A|)$ and $m_{S_t}(ABC)$ by
$|ABC|(2n - |ABC|)$, allowing for us to obtain a bound on
$m_{S_{t}}(AC)$ in terms of $a, b, c$, namely, 
\begin{equation}
  m_{S_{t}}(AC) \le \frac{c(a+b+c) (2n - a-b-c)}{b+c} + \frac{b
  a(2n-a)}{b+c}.
  \label{eq:abccomplicated}
\end{equation}
By algebraic manipulation, \eqref{eq:abccomplicated} reduces to
\begin{equation}
  \label{eq:TO_goodsstsuffminuscb}
  m_{S_t}(AC) \le  k(2n-k)-cb.
\end{equation}
The transformation from \eqref{eq:ACbound} to
\eqref{eq:TO_goodsstsuffminuscb} might at first seems somewhat
mysterious; in the full version of the proof we also give an
alternative version of the transformation that elucidates some of
the underlying combinatorial structure.

Combined with \eqref{eq:TO_bplusmAC}, and the fact $c>0$ (which
follows from $b > 0$), \eqref{eq:TO_goodsstsuffminuscb} implies
that $m_{S_{t + 1}}(AC) \le k(2n - k)$ and thus that the average
fill of the $k$ fullest cups in state $S_{t+1}$ is at most
$2n-k$, as desired.
\end{proof}
\subsection{Oblivious Lower Bound}

In \cref{sec:oblivious} we consider what happens if the filler is an
\defn{oblivious} adversary, meaning that the filler cannot see what
the emptier does at each step. The emptier, in turn, is permitted to
use randomization in order to make its behavior unpredictable to the
filler. We focus on randomized emptying algorithms that satisfy the
so-called \defn{$\Delta$-greedy-like} property: the emptier never
empties from a cup $c$ over another cup $c'$ whose fill is more than
$\Delta$ greater than the fill of $c$.

The main theorem in \cref{sec:oblivious} gives an oblivious filling
strategy that achieves backlog $\Omega(n^{1 - \epsilon})$ against any
$\Delta$-greedy-like emptier for any $\Delta \in \Omega(1)$ (or, more
precisely, any $\Delta \le \frac{1}{128} \log \log \log n$).

\begin{theorem}[\cref{thm:obliviousPoly}]
  There is an oblivious filling strategy for the
  variable-processor cup game on $N$ cups that achieves backlog
  at least $\Omega(N^{1-\varepsilon})$ for any constant $\varepsilon
  >0$ in running time $2^{\polylog(N)}$ with probability at least
  $1-2^{-\polylog(N)}$ against a $\Delta$-greedy-like emptier
  with $\Delta \le \frac{1}{128} \log\log\log N$.
  \label{thm:oblivious-main}
\end{theorem}

Note that Theorem \ref{thm:oblivious-main} uses $N$ for the number of
cups rather than using $n$. When describing the recursive strategy
that the filler uses, we will use $N$ to denote the true total number
of cups, and $n$ to denote the number of cups within the recursive
subproblem currently being discussed.

The filling strategy used in Theorem \ref{thm:oblivious-main} is
closely related the adaptive filling strategy described in Section
\ref{sec:technicaladaptive}. The fact that the filling strategy must
now be \emph{oblivious}, however, introduces several new technical
obstacles.

\paragraph{Problem 1: Distinguishing between neglected and non-neglected phases}
Recall that the Amplification Lemma in Section
\ref{sec:technicaladaptive} proceeds in phases, where the filler
behaves differently at the end of each phase depending on whether or
not the emptier ever neglected the anchor set $A$ during that
phase. If the filler is oblivious, however, then it cannot detect
which phases are neglected.

To solve this problem, the first thing to notice is that the
\emph{total} number of times that the emptier can neglect the anchor
set (within a given recursive subproblem of the Amplification Lemma)
is, without loss of generality, at most $N^2$. Indeed, if the emptier
neglects the anchor set more than $N^2$ times, then the total amount
of water in cups $A$ will be at least $N^2$. Since the amount of water
in the system as a whole is non-decreasing, there will subsequently
always be at least one cup in the system with fill $N$ or larger, and
thus the filler's strategy trivially achieves backlog $N$.

Assume that there are at most $N^2$ phases that the emptier
neglects. The filler does not know which phases these are, and
the filler does not wish to ever move a cup from the non-anchor
set to the anchor set during a phase that the emptier neglected
(since, during such a phase, there is no guarantee on the amount
of water in the cup from $B$). To solve this problem, we increase
the total number of phases in the Amplification Lemma to be some
very large number $M = 2^{\polylog N}$, and we have the filler
select $|A|$ random phases at the end of which to move a cup from
the non-anchor set to the anchor set. With high probability, none
of the $|A|$ phases that the filler selects are neglected by the
emptier.

\paragraph{Problem 2: Handling the probability of failure}
Because the filler is now oblivious (and the emptier is randomized)
the guarantee offered by the filling strategy is necessarily
probabilistic. This makes the Amplification Lemma somewhat more
difficult, since each application of $\alg(f)$ now has some
probability of failure.

We ensure that the applications of $\alg(f)$ succeed with such
high probability that we can ignore the possibility of any of
them failing on phases when we need them to succeed. This
necessitates making sure that the base-case construction
$\alg(f_0)$ succeeds with very high probability.

Fortunately, we can take a base-case construction $\alg_0$ that
succeeds with only constant (or even sub-constant) probability,
and perform an Amplification-Lemma-like construction in order to
obtain a new filling strategy $\alg_1$ that achieves slightly
\emph{smaller} backlog, but that has a very high probability of
succeeding.

To construct $\alg_1$, we begin by performing the
Amplification-Lemma construction on $\alg_0$, but without
recursing after the final phase. Even though many of the
applications of $\alg_0$ may fail, with high probability at least
one application succeeds. This results in some cup $c_*$ in $A$
having high fill. Unfortunately, the filler does not know which
cup has high fill, so it cannot simply take $c_*$. What the
filler can do, however, is select some cup $c$, decrease the
number of processors to $1$, and then spend a large number of
steps simply placing $1$ unit of water into cup $c$ in each step.
By the $\Delta$-greediness of the emptier, the emptier is
guaranteed to focus on emptying from cup $c_*$ (rather than cup
$c$) until $c$ attains large fill. This allows for the filler to
obtain a cup $c$ that the filler \emph{knows} contains a large
amount of water (with high probability). We use this approach to
construct a base-case filling strategy $\alg(f_0')$ that succeeds
with high probability.

\paragraph{Problem 3: Non-flat starting states}
The next problem that we encounter is that, at the beginning of
any given phase, the cups in the non-anchor set may not start off
with equal heights. Instead, some cups may contain very large
amounts of water while others contain very small (and even
negative) amounts of water\footnote{Recall that, in order for our
lower-bound construction to be able to call itself recursively,
we must analyze the construction in the negative-fill version of
the variable-processor cup game.}. This is not a problem for an
adaptive filler, since the filler knows which cups contain
small/large amounts of water, but it is a problem for an
oblivious filler.

To avoid the scenario in which the cups in $B$ are highly unequal, we
begin each phase by first performing a \defn{flattening construction}
on the cups $B$, which causes the cups in $B$ to all have roughly
equal fills (up to $\pm O(\Delta)$). The flattening construction uses
the fact that the emptier is $\Delta$-greedy-like to ensure that cups
which are overly full get flattened out by the emptier.

\paragraph{Putting the pieces together}
By combining the ideas above, as well as handling other 
issues that arise (e.g., one must be careful to ensure that the
average fills of $A$ and $B$ do not drift apart in unpredictable
ways), one can prove Theorem \ref{thm:oblivious-main}. In Section
\ref{sec:oblivious}, we give the full proof, which is by far the most
technically involved result in the paper.

\section{Adaptive Filler Lower Bound}\label{sec:adaptive}

In this section we give a $2^{\polylog n}$-time filling strategy
that achieves backlog $n^{1 - \varepsilon}$ for any positive
constant $\varepsilon$. 
We also give a $O(n!)$-time filling strategy that achieves
backlog $\Omega(n)$.

We begin with a trivial filling strategy that we call
\defn{$\trivalg$} that gives backlog at least $1/2$ when applied
to at least $2$ cups.
\begin{proposition}
  \label{prop:adaptiveBase}
  Consider an instance of the negative-fill $1$-processor cup
  game on $n$ cups, and let the cups start in any state with
  average fill is $0$. If $n\ge 2$, there is an $O(1)$-step
  adaptive filling strategy $\trivalg$ that achieves backlog at
  least $1/2$. If $n=1$, $\trivalg$ achieves backlog $0$ in
  running time $0$.
\end{proposition}
\begin{proof}
  If $n=1$, $\trivalg$ does nothing and achieves backlog $0$; for
  the rest of the proof we consider the case $n\ge 2$.

  Let $a$ and $b$ be the fullest and second fullest cups in the in
  the starting configuration, and let their initial fills be
  $\fil(a) = \alpha, \fil(b) = \beta$. 
  If $\alpha\ge 1/2$ the filler need not do anything, the desired
  backlog is already achieved.
  Otherwise, if $\alpha \in [0, 1/2]$, the filler places
  $1/2-\alpha$ fill into $a$ and $1/2 + \alpha$ fill into $b$
  (which is possible as both fills are in $[0,1]$, and they sum
  to $1$). Since $\alpha + \beta \ge 0$ we have $\beta \ge -\alpha$.
  Clearly $a$ and $b$ now both have fill at least $1/2$.
  The emptier cannot empty from both $a$ and $b$ as $p=1$, so
  even after the emptier empties from a cup we still have backlog
  $1/2$, as desired.
\end{proof}

Next we prove the \defn{Amplification Lemma}, which takes as
input a filling strategy $\alg(f)$ and outputs a new filling
strategy $\alg(f')$ that we call the \defn{amplification} of
$\alg(f)$. $\alg(f')$ is able to achieve higher fill than
$\alg(f)$; in particular, we will show that by starting with a
filling strategy $\alg(f_0)$ for achieving constant backlog and
then forming a sufficiently long sequence of filling strategies
$\alg(f_0), \alg(f_1), \ldots, \alg(f_{i_*})$ with
$\alg(f_{i+1})$ the amplification of $\alg(f_i)$, we 
get a filling strategy for achieving $\poly(n)$ backlog.

\begin{lemma}[Adaptive Amplification Lemma]
  \label{lem:adaptiveAmplification}
  Let $\delta\in(0,1/2]$ be a parameter.
  Let $\alg(f)$ be an adaptive filling strategy that 
  achieves backlog $f(n) < n$ in the negative-fill variable-processor cup game
  on $n$ cups in running time $T(n)$ starting from any initial
  cup state where the average fill is $0$.

  Then there exists an adaptive filling strategy $\alg(f')$ that
  achieves backlog $f'(n)$ satisfying 
  $$f'(n) \ge (1-\delta)f(\floor{(1-\delta)n}) + f(\ceil{\delta n}) $$
  and $f'(n) \ge f(n)$
  in the negative-fill variable-processor cup game on $n$ cups in running time 
  $$T'(n) \le n\ceil{\delta n} \cdot T(\floor{(1-\delta)n}) + T(\ceil{\delta n})$$
  starting from any initial cup state where the average fill is $0$.
\end{lemma}
\begin{proof}
  Let $n_A = \ceil{\delta n}, n_B = n-n_A = \floor{(1-\delta) n}$.

  The filler defaults to using $\alg(f)$ if 
  $$f(n) \ge (1-\delta)f(n_B) + f(n_A).$$ 
  In this case using $\alg(f)$ achieves the desired backlog in
  the desired running time. In the rest of the proof, we describe
  our strategy for the case where we cannot simply use $\alg(f)$ to
  achieve the desired backlog. 

  Let $A$, the \defn{anchor set}, be initialized to consist of
  the $n_A$ fullest cups, and let $B$ the
  \defn{non-anchor set} be initialized to consist of the rest of
  the cups (so $|B| = n_B$). Let $h = (1-\delta)f(n_B).$

  The filler's strategy can be summarized as follows: \\
  \textbf{Step 1:} Make $\mu(A) \ge h$ by using $\alg(f)$
  repeatedly on $B$ to achieve cups with fill at least $\mu(B) +
  f(n_B)$ in $B$ and then swapping these into $A$. While doing
  this the filler always places $1$ unit of fill in each anchor
  cup.\\
  \textbf{Step 2:} Use $\alg(f)$ once on $A$ to obtain some cup
  with fill $\mu(A)+f(n_A)$.\\
  Note that in order to use $\alg(f)$ on subsets of the cups the
  filler will need to vary $p$.

  We now describe how to achieve Step 1, which is
  complicated by the fact that the emptier may attempt to
  prevent the filler from achieving high fill in a cup
  in $B$.

  The filling strategy always places $1$ unit of water in each
  anchor cup. This ensures that no cups in the anchor set ever
  have their fill decrease. If the emptier wishes to keep the
  average fill of the anchor cups from increasing, then emptier
  must empty from every anchor cup on each step. If the emptier
  fails to do this on a given round, then we say that the emptier
  has \defn{neglected} the anchor cups. 

  We say that the filler \defn{applies} $\alg(f)$ to $B$ if it
  follows the filling strategy $\alg(f)$ on $B$ while placing $1$
  unit of water in each anchor cup. An application of $\alg(f)$
  to $B$ is said to be \defn{successful} if $A$ is never
  neglected during the application of $\alg(f)$ to $B$. The
  filler uses a series of phases that we call
  \defn{swapping-processes} to achieve the desired average fill
  in $A$. In a swapping-process, the filler repeatedly applies
  $\alg(f)$ to $B$ until a successful application occurs, and
  then takes the cup generated by $\alg(f)$ within $B$ on this
  successful application with fill at least $\mu(B) + f(|B|)$ and
  swaps it with the least full cup in $A$ so long as the swap
  increases $\mu(A)$. If the average fill in $A$ ever
  reaches $h$, then the algorithm immediately halts (even if it
  is in the middle of a swapping-process) and is complete.

  We give pseudocode for the filling strategy in
  \cref{alg:adaptiveAmplification}.

\begin{algorithm}
  \caption{Adaptive Amplification (Step 1)}
  \label{alg:adaptiveAmplification}
  \begin{algorithmic}
    \State \textbf{Input:} $\alg(f), \delta, $ set of $n$ cups
    \State \textbf{Output:} Guarantees that $\mu(A) \ge h$
    \State
    \State $A \gets n_A$ fullest cups, $B \gets $ rest of the cups
    \State Always place $1$ fill in each cup in $A$
    \While{$\mu(A) < h$} \Comment Swapping-Processes
    \State Immediately \textbf{exit} this loop if ever $\mu(A) \ge h$ 
      \State successful $\gets $ false
      \While{not successful} 
      \State Apply $\alg(f)$ to $B$, $\alg(f)$ gives cup $c$
        \If{$\fil(c) \ge h$}
          \State successful $\gets$ true
        \EndIf
      \EndWhile
      \State Swap $c$ with least full cup in $A$
    \EndWhile
  \end{algorithmic}
\end{algorithm}
  
  Note that $$\mu(A) \cdot |A| + \mu(B)\cdot |B| = \mu(AB) \ge 0,$$
  as $\mu(AB)$ starts as $0$, but could become positive if the
  emptier skips emptyings.
  Thus we have 
  $$\mu(A) \ge - \mu(B) \cdot
  \frac{\floor{(1-\delta)n}}{\ceil{\delta n}} \ge -
  \frac{1-\delta}{\delta} \mu(B).$$ Thus, if at any
  point $B$ has average fill lower than $-h \cdot
  \delta/(1-\delta)$, then $A$ has average fill at least $h$, so
  the algorithm is finished. Thus we can assume in our analysis that
  \begin{equation}
    \mu(B) \ge -h\cdot\delta/(1-\delta).
  \label{eq:Batleast}
  \end{equation}
  We will now show that the filler
  applies $\alg(f)$ to $B$ at most $h n_A$ total times. 
  Each time the emptier neglects the anchor set, the mass of the
  anchor set increases by $1$. If the emptier neglects the anchor set $h
  n_A$ times, then the average fill in the anchor set increases by
  $h$. Since $\mu(A)$ started as at least $0$, and
  since $\mu(A)$ never decreases (note in particular that cups are only
  swapped into $A$ if doing so will increase $\mu(A)$), an
  increase of $h$ in $\mu(A)$ implies that $\mu(A) \ge h$, as
  desired.  

  Consider the fill of a cup $c$ swapped into $A$ at the end of a
  swapping-process. Cup $c$'s fill is at least $\mu(B) + f(n_B)$,
  which by \eqref{eq:Batleast} is at least
  $$-h \cdot \frac{\delta}{1-\delta} + f(n_B) = (1-\delta)f(n_B) = h.$$ 
  Thus the algorithm for Step 1 succeeds within $|A|$
  swapping-processes, since at the end of the $|A|$-th swapping
  process either every cup in $A$ has fill at least $h$, or the
  algorithm halted before $|A|$ swapping-processes because it
  already achieved $\mu(A) \ge h$. 
  
  After achieving $\mu(A) \ge h$, the filler performs Step 2,
  i.e. the filler applies $\alg(f)$ to $A$, and hence achieves a
  cup with fill at least $$\mu(A) + f(|A|) \ge (1-\delta)f(n_B) +
  f(n_A),$$ as desired.

  Now we analyze the running time of the filling strategy
  $\alg(f')$. First, recall that in Step 1 $\alg(f')$ calls
  $\alg(f)$ on $B$, which has size $n_B$, as many as $h n_A$ times.
  Because we mandate that $h < n$, Step 1 contributes no more
  than $(n\cdot n_A) \cdot T(n_B)$ to the running time.
  Step 2 requires applying $\alg(f)$ to $A$, which has size
  $n_A$, once, and hence
  contributes $T(n_A)$ to the running time. Summing these we have
  $$T'(n) \le n \cdot n_A \cdot T(n_B) + T(n_A).$$

\end{proof}

We next show that by recursively using the Amplification Lemma we
can achieve backlog $n^{1 - \varepsilon}$.
\begin{theorem}
  \label{thm:adaptivePoly}
  There is an adaptive filling strategy for the variable-processor cup game on
  $n$ cups that achieves backlog $\Omega(n^{1-\varepsilon})$ for any constant
  $\varepsilon > 0$ of our choice in running time $2^{O(\log^2 n)}$.
\end{theorem}
\begin{proof}
  Take constant $\varepsilon \in (0,1/2)$. Let $c, \delta$ be
  constants that will be chosen (later) as functions of
  $\varepsilon$ satisfying $c\in (0,1), 0 < \delta \ll 1/2$.
  We show how to achieve backlog at least $cn^{1-\varepsilon}-1$.

 Let $\alg(f_0)=\trivalg$, the algorithm given by
 \cref{prop:adaptiveBase}; recall $\trivalg$ achieves backlog $f_0(k)
 \ge 1/2$ for all $k \ge 2$, and $f_0(1) = 0$. 
  Next, using the Amplification Lemma we recursively construct
  $\alg(f_{i+1})$ as the amplification of $\alg(f_{i})$ for $i\ge 0$. 
  Define a sequence $g_i$ with 
  $$ g_i = \begin{cases}
    \ceil{16/\delta},  & i = 0,\\
    \floor{g_{i-1}/(1-\delta)} & i \ge 1.
  \end{cases}$$
  We claim the following regarding this construction:
  \begin{clm}
    \label{clm:fikinduction}
    For all $i\ge0$,
    \begin{equation}
      f_i(k) \ge ck^{1-\varepsilon}-1\,\, \text{ for all }\,\, k\in [g_{i}].
    \label{eqn:fikinduction}
    \end{equation}
  \end{clm}
  \begin{proof}
  We prove \cref{clm:fikinduction} by induction on $i$.
  For $i=0$, the base case, \eqref{eqn:fikinduction} can
  be made true by taking $c$ sufficiently small; in particular,
  taking $c<1$ makes \eqref{eqn:fikinduction} hold for $k = 1$,
  and, as $g_0 > 2$, taking $c$ small enough to make $c
  g_0^{1-\varepsilon} -1 \le f_0(g_0) = 1/2$ makes
  \eqref{eqn:fikinduction} hold for $k\in [2, g_0]$ by
  monotonicity of $k \mapsto ck^{1-\varepsilon}-1$\footnote{Note
  that it is important here that $\varepsilon$ and
  $\delta$ are constants, that way $c$ is also a constant.}.

  As our inductive hypothesis we assume
  \eqref{eqn:fikinduction} for $f_i$; we aim to show that 
  \eqref{eqn:fikinduction} holds for $f_{i+1}$. Note that, by
  design of $g_i$, if $k \le g_{i+1}$ then $\floor{k\cdot (1-\delta)} \le g_i$.
  Consider any $k\in [g_{i+1}]$. First we deal with the trivial
  case where $k \le g_0$. In this case 
  $$f_{i+1}(k) \ge f_i(k) \ge \cdots \ge f_0(k) \ge ck^{1-\varepsilon} -1.$$
  Now we consider the case where $k \ge g_0$.
  Since $f_{i+1}$ is the amplification of $f_i$ we have
  $$f_{i+1}(k) \ge (1-\delta) f_i(\floor{(1-\delta)k}) + f_i(\ceil{\delta k}).$$
  By our inductive hypothesis, which applies as $\ceil{\delta k}\le g_i, \floor{k\cdot (1-\delta)} \le g_i$, we have
  $$f_{i+1}(k) \ge (1-\delta) (c \cdot\floor{(1-\delta)k}^{1-\varepsilon}-1) + c\ceil{\delta k}^{1-\varepsilon} - 1. $$
  Dropping the floor and ceiling, incurring a $-1$ for dropping the floor, we have
  $$f_{i+1}(k) \ge (1-\delta) (c \cdot ((1-\delta)k-1)^{1-\varepsilon}-1) + c (\delta k)^{1-\varepsilon} - 1.$$
  Because $(x-1)^{1-\varepsilon} \ge x^{1-\varepsilon} -1$, as $x\mapsto x^{1-\varepsilon}$ is a sub-linear
  sub-additive function, we have 
  $$f_{i+1}(k) \ge (1-\delta) c \cdot (((1-\delta)k)^{1-\varepsilon}-2) + c(\delta k)^{1-\varepsilon}-1.$$
  Moving the $ck^{1-\varepsilon}$ to the front we have
  $$ f_{i+1}(k) \ge ck^{1-\varepsilon} \cdot\left((1-\delta)^{2-\varepsilon} + \delta^{1-\varepsilon} - \frac{2(1-\delta)}{k^{1-\varepsilon}} \right) -1.$$
  Because $(1-\delta)^{2-\varepsilon} \ge 1-(2-\varepsilon)\delta$, a fact called Bernoulli's Identity, we have
  $$f_{i+1}(k) \ge ck^{1-\varepsilon} \cdot\left(1-(2-\varepsilon)\delta + \delta^{1-\varepsilon} - \frac{2(1-\delta)}{k^{1-\varepsilon}} \right)-1.$$
  Of course $-2(1-\delta) \ge -2$, so 
  $$f_{i+1}(k) \ge ck^{1-\varepsilon} \cdot\left(1-(2-\varepsilon)\delta + \delta^{1-\varepsilon} - \frac{2}{k^{1-\varepsilon}} \right) -1.$$
  Because 
  $$\frac{-2}{k^{1-\varepsilon}} \ge \frac{-2}{g_0^{1-\varepsilon}} \ge -2(\delta/16)^{1-\varepsilon} \ge -\delta^{1-\varepsilon}/2,$$ 
  which follows from our choice of $g_0 = \ceil{16/\delta}$ and the restriction
  $\varepsilon<1/2$, we have
  $$f_{i+1}(k) \ge ck^{1-\varepsilon}
  \cdot\left(1-(2-\varepsilon)\delta + \delta^{1-\varepsilon} - \delta^{1-\varepsilon}/2 \right)-1.$$
  Finally, combining terms we have
  $$f_{i+1}(k) \ge  ck^{1-\varepsilon} \cdot\left(1-(2-\varepsilon)\delta + \delta^{1-\varepsilon}/2\right)-1. $$
  Because $\delta^{1-\varepsilon}$ dominates $\delta$ for
  sufficiently small $\delta$, there is a choice of
  $\delta=\Theta(1)$ such that 
  $$1-(2-\varepsilon)\delta + \delta^{1-\varepsilon}/2 \ge 1.$$ 
  Taking $\delta$ to be this
  small we have,
  $$f_{i+1}(k) \ge ck^{1-\varepsilon}-1,$$
  completing the proof. We remark that the choices of $c, \delta$
  are the same for every $i$ in the inductive proof, and depend
  only on $\varepsilon$. 
  \end{proof}

  To complete the proof, we will show that $g_i$ grows
  exponentially in $i$. This implies that there exists $i_* \le
  O(\log n)$ such that $g_{i_*} \ge n$, and hence we have an
  algorithm $\alg(f_{i_*})$ that achieves backlog
  $cn^{1-\varepsilon}-1$ on $n$ cups, as desired.
  
  We lower bound the sequence $g_i$ with another sequence $g_i'$
  defined as 
  $$g_i'=\begin{cases}
    4/\delta, & i=0\\
    g_{i-1}' / (1-\delta) -1, & i> 0.
  \end{cases}$$
  Solving this recurrence, we find 
  $$g_i' = \frac{4-(1-\delta)^2}{\delta} \frac{1}{(1-\delta)^i}
  \ge \frac{1}{(1-\delta)^i},$$
  which clearly exhibits exponential growth. 
  In particular, let $i_* = \ceil{\log_{1/(1-\delta)} n}$. Then,
  $g_{i_*} \ge g_{i_*}' \ge n,$ as desired.

  Let the running time of $f_i(n)$ be $T_i(n)$. From the
  Amplification Lemma we have following recurrence bounding
  $T_i(n)$:
  \begin{align*}
    T_i(n) &\le n\ceil{\delta n} \cdot T_{i-1}(\floor{(1-\delta)n}) +
  T_{i-1}(\ceil{\delta n}) \\
  &\le 2n^2T_{i-1}(\floor{(1-\delta)n}).
  \end{align*}
  It follows that $\alg(f_{i_*})$, recalling that $i_* \le O(\log n)$, has running time
  $$T_{i_*}(n) \le (2n^2)^{O(\log n)} \le 2^{O(\log^2 n)},$$
  as desired.

\end{proof}

Now we provide a construction that can achieve backlog $\Omega(n)$
in very long games. The construction can be interpreted as the same
argument as in \cref{thm:adaptivePoly} but with an extremal setting of
$\delta$ to $\Theta(1/n)$\footnote{Or more precisely, setting
$\delta$ in each level of recursion to be $\Theta(1 / n)$, where
$n$ is the subproblem size; note in particular that $\delta$
changes between levels of recursion, which was not the case in
the proof of \cref{thm:adaptivePoly}.}.

\begin{theorem}
  \label{thm:factorialTimeAlg}
  There is an adaptive filling strategy that
  achieves backlog $\Omega(n)$ in time $O(n!)$.
\end{theorem}
\begin{proof}
  First we construct a slightly stronger version of $\trivalg$
  that achieves backlog $1$ on $n \ge n_0=8$ cups, instead of
  just backlog $1/2$; this simplifies the analysis.
  \begin{clm}
    There is a filling algorithm $\trivalg_2$ that achieves
    backlog at least $1$ on $n_0 = 8$ cups.
  \end{clm}
  \begin{proof}
    Let $\trivalg_1$ be the amplification of $\trivalg$ using
    $\delta = 1/2$. On $4$ cups $\trivalg_1$ achieves backlog at
    least $(1/2)(1/2)+1/2 = 3/4$.
    Let $\trivalg_2$ be the amplification of $\trivalg_1$ using
    $\delta = 1/2$. On $8$ cups $\trivalg_2$ achieves backlog at
    least $(1/2)(3/4) + 3/4 \ge 1$.
  \end{proof}

  Let $\alg(f_0)=\trivalg_2$; we have $f_0(k) \ge 1$ for all $k
  \ge n_0$. For $i > 0$ we construct $\alg(f_{i})$ as the
  amplification of $\alg(f_{i-1})$ using the Amplification Lemma
  with parameter $\delta = 1/(i+1)$. 

  We claim the following regarding this construction:
  \begin{clm}
    \label{clm:yayactuallygetn}
  For all $i\ge 0$,
  \begin{equation}
    \label{eq:omegaNpfinduction}
    f_i((i+1)n_0) \ge \sum_{j=0}^i \left(1-\frac{j}{i+1}\right).
  \end{equation}
  \end{clm}
  \begin{proof}
  We prove \cref{clm:yayactuallygetn} by induction on $i$. When
  $i=0$, the base case, \eqref{eq:omegaNpfinduction} becomes
  $f_{0}(n_0) \ge 1$ which is true. Assuming
  \eqref{eq:omegaNpfinduction} for $f_{i-1}$, we now show
  \eqref{eq:omegaNpfinduction} holds for $f_{i}$.
  Because $f_{i}$ is the amplification of $f_{i-1}$ with $\delta = 1/(i+1)$, we have by the Amplification Lemma
  $$f_{i}((i+1)\cdot n_0) \ge \left(1 - \frac{1}{i+1}\right) f_{i-1}(i\cdot n_0) + f_{i-1}(n_0).$$
  Since $f_{i-1}(n_0) \ge f_0(n_0) \ge 1$ we have
  $$f_{i}((i+1)\cdot n_0) \ge \left(1 - \frac{1}{i+1}\right) f_{i-1}(i\cdot n_0) + 1.$$
  Using the inductive hypothesis we have
  $$f_{i}((i+1)\cdot n_0) \ge \left(1 - \frac{1}{i+1}\right)\sum_{j=0}^{i-1} \left(1-\frac{j}{i}\right) + 1.$$
  Note that 
  $$\left(1 - \frac{1}{i+1}\right)\cdot
  \left(1-\frac{j}{i}\right) = \frac{i}{i+1} \cdot \frac{i-j}{i}
  = \frac{i-j}{i+1} = 1 - \frac{j+1}{i+1}.$$
  Thus we have the desired bound:
  $$f_{i}((i+1)\cdot n_0) \ge \sum_{j=1}^{i}
  \left(1-\frac{j}{i+1}\right) + 1 = \sum_{j=0}^{i}
  \left(1-\frac{j}{i+1}\right).$$
  \end{proof}

  Let $i_* = \floor{n/n_0}-1$, which by design satisfies
  $(i_*+1)n_0 \le n$. By \cref{clm:yayactuallygetn} we have
  $$f_{i_*}((i_*+1)\cdot n_0) \ge \sum_{j=0}^{i_*} \left(1 -
  \frac{j}{i_*+1} \right) = i_*/2 + 1.$$ As $i_* = \Theta(n)$, we
  have thus shown that $\alg(f_{i_*})$ can achieve backlog
  $\Omega(n)$ on $n$ cups.

  Let $T_i$ be the running time of $\alg(f_i)$.
  The recurrence for the running running time of $f_{i_*}$ is 
  $$T_i(n) \le n \cdot n_0T_{i-1}(n-n_0)+O(1).$$
  Clearly $T_{i_*}(n) \le O(n!)$.

\end{proof}

\section{Upper Bound}\label{sec:upperBound}

In this section we analyze the \defn{greedy emptier}, which
always empties from the $p$ fullest cups. We prove in
\cref{cor:upperbound} that the greedy emptier prevents backlog
from exceeding $O(n)$. In order to analyze the greedy emptier, we
establish a system of invariants that hold at every step of the
game. 

The main result of the section is the following theorem\footnote{
  Recall that we use $\mu_S(X)$ and $m_S(X)$ to denote the average
  fill and the mass, respectively, of a set of cups $X$ at state
  $S$. In some parts of the paper we can omit the subscript indicating
  the state at which the properties are measured, since the state is
  clear. However, in \cref{sec:upperBound} it is necessary to make the
  state $S$ explicit in the notation.}.
\begin{theorem}
  \label{thm:invariant}
  In the variable-processor cup game on $n$ cups, the greedy
  emptier maintains, at every step $t$,
  the invariants
  \begin{equation}
    \label{eq:invariants}
      \mu_{S_t}(S_t([k])) \le 2n-k
  \end{equation}
  for all  $k \in [n]$.
\end{theorem}

By applying \cref{thm:invariant} to the case of $k = 1$, we arrive at a bound on backlog:
\begin{corollary}
  In the variable-processor cup game on $n$ cups, the greedy
  emptying strategy never lets backlog exceed $O(n)$.
  \label{cor:upperbound}
\end{corollary}

\begin{proof}[Proof of \cref{thm:invariant}]
We prove the invariants by induction on $t$.
The invariants hold trivially for $t=1$ (the base case for the inductive proof): 
the cups start empty so $\mu_{S_1}(S_1([k])) = 0 \le 2n-k$ for all $k \in [n]$.

Fix a round $t \ge 1$, and any $k \in [n]$. We assume the invariant for all
values of $k' \in[n]$ for state $S_t$ (we will only explicitly use two of the
invariants for each $k$, but the invariants that we need depend on the
choice of $p_t$ by the filler) and show that
the invariant on the $k$ fullest cups holds on round $t+1$, i.e. that
$$\mu_{S_{t+1}}(S_{t+1}([k])) \le 2n-k.$$

Note that because the emptier is greedy it always empties from the cups
$I_t([p_t])$. Let $A$, with $a=|A|$, be $A = I_t([\min(k, p_t)]) \cap
S_{t+1}([k])$; $A$ consists of the cups that are among the $k$ fullest cups in
$I_t$, are emptied from, and are among the $k$ fullest cups in $S_{t+1}$. Let
$B$, with $b=|B|$, be $I_t([\min(k, p_t)]) \setminus A$; $B$ consists of the
cups that are among the $k$ fullest cups in state $I_t$, are emptied from, and
are not among the $k$ fullest cups in $S_{t+1}$. Let $C = I_t(a+b+[k-a])$, with
$c=k-a = |C|$; $C$ consists of the cups with ranks $a + b + 1, \ldots, k + b$
in state $I_t$. The set $C$ is defined so that $S_{t+1}([k]) = AC$, since once
the cups in $B$ are emptied from, the cups in $B$ are not among the $k$ fullest
cups, so cups in $C$ take their places among the $k$ fullest cups.

Note that $k-a \ge 0$ as $a+b \le k$, and also $|ABC| = k+b \le n$, because by
definition the $b$ cups in $B$ must not be among the $k$ fullest cups in state
$S_{t+1}$ so there are at least $k+b$ cups. 
Note that $a + b = \min(k, p_t)$. We also have that $A = I_t([a])$ and $B =
I_t(a+[b])$, as every cup in $A$ must have higher fill than all cups in $B$ in
order to remain above the cups in $B$ after $1$ unit of water is removed from
all cups in $AB$.

We now establish the following claim, which we call the \defn{interchangeability of cups}:
\begin{clm}
  \label{clm:interchangable}
  There exists a cup state $S_t'$ such that: (a) $S_t'$ satisfies the
  invariants \eqref{eq:invariants}, (b) $S_t'(r) = I_t(r)$ for all ranks
  $r\in[n]$, and (c) the filler can legally place water into cups in order to
  transform $S_t'$ into $I_t$. 
\end{clm}
\begin{proof}
  Fix $r \in [n]$. We will show that $S_t$ can be transformed into a state
  $S_t^r$ by relabelling only cups with ranks in $[r]$ such that (a) $S_t^r$
  satisfies the invariants \eqref{eq:invariants}, (b) $S_t^r([r]) = I_t([r])$
  and (c) the filler can legally place water into cups in order to transform
  $S_t^r$ into $I_t$.

Say there are cups $x, y$ with $x\in S_t([r]) \setminus I_t([r]), y \in
 I_t([r])\setminus S_t([r])$. Let the fills of cups $x,y$ at state $S_t$
 be $f_x, f_y$; note that 
 \begin{equation}
     f_x > f_y.
     \label{eq:fxfy}
 \end{equation} Let the amount of fill that the filler
 adds to these cups be $\Delta_x, \Delta_y \in [0,1]$; note that 
 \begin{equation}
 f_x +\Delta_x <f_y + \Delta_y.
 \label{eq:fxdxfydy}
 \end{equation}
 
Define a new state $S_t'$ where cup $x$ has fill $f_y$ and cup $y$ has fill
$f_x$. Note that the filler can transform state $S_t'$ into state $I_t$ by
placing water into cups as before, except changing the amount of water placed
into cups $x$ and $y$ to be  $f_x-f_y+\Delta_x$ and $f_y-f_x + \Delta_y$,
respectively.

In order to verify that the transformation from $S_t'$ to $I_t$ is a valid step
for the filler, one must check three conditions. First, the amount of water
placed by the filler is unchanged: this is because $(f_x-f_y + \Delta_x) +
(f_y-f_x+\Delta_y) = \Delta_x + \Delta_y$. Second, the fills placed in cups $x$
and $y$ are at most $1$: this is because $f_x-f_y+\Delta_x<\Delta_y \le 1$ (by
\eqref{eq:fxdxfydy}) and $f_y-f_x + \Delta_x < \Delta_x \le 1$ (by
\eqref{eq:fxfy}). Third, the fills placed in cups $x$ and $y$ are non-negative:
this is because $f_x-f_y + \Delta_x > \Delta_x \ge 0$ (by \eqref{eq:fxfy})
and $f_y-f_x+\Delta_y > \Delta_x \ge 0$ (by
\eqref{eq:fxdxfydy}). 

We can repeatedly apply this process to swap each cup in $I_t([r])\setminus
S_t([r])$ into being in $S_t'([r])$. At
the end of this process we will have some state $S_t^r$ for which
$S_t^r([r]) = I_t([r])$. Note that $S_t^r$ is simply a relabeling of $S_t$,
hence it must satisfy the same invariants \eqref{eq:invariants} satisfied by
$S_t$. Further, $S_t^r$ can be transformed into $I_t$ by a valid filling step.

Now we repeatedly apply this process, in descending order of ranks. 
In particular, we have the following process: create a sequence of states by
starting with $S_t^{n-1}$, and to get to state $S_t^{r}$ from state $S_t^{r+1}$
apply the process described above. 
Note that $S_t^{n-1}$ satisfies $S_t^{n-1}([n-1]) = I_t([n-1])$ and thus also
$S_t^{n-1}(n) = I_t(n)$.
If $S_t^{r+1}$ satisfies $S_t^{r+1}(r') = I_t(r')$ for all $r'>r+1$ then
$S_t^r$ satisfies $S_t^r(r') = I_t(r')$ for all $r > r$, because the transition
from $S_t^{r+1}$ to $S_t^r$ has not changed the labels of any cups with ranks
in $(r+1, n]$, but the transition does enforce $S_t^r([r]) = I_t([r])$, and
consequently $S_t^r(r+1) = I_t(r+1)$.
We continue with the sequential process until arriving at state $S_t^1$ in
which we have $S_t^1(r) = I_t(r)$ for all $r$.
Throughout the process each $S_t^r$ has satisfied the invariants
\eqref{eq:invariants}, so $S_t^1$ satisfies the invariants
\eqref{eq:invariants}. Further, throughout the process from each $S_t^r$ it is
possible to legally place water into cups in order to transform $S_t^r$ into
$I_t$.

Hence $S_t^1$ satisfies all the properties desired, and the proof of 
\cref{clm:interchangable} is complete.

\end{proof}

\cref{clm:interchangable} tells us that we may assume without loss of
generality that $S_t(r) = I_t(r)$ for each rank $r \in [n]$. We will make
this assumption for the rest of the proof. 

In order to complete the proof of the theorem, we break it into three cases. 

\begin{clm}
  If some cup in $A$ zeroes out in round $t$, then the invariant
  $\mu_{S_{t+1}}(S_{t+1}([k])) \le 2n-k$ holds.
\end{clm}
\begin{proof}
  Say a cup in $A$ zeroes out in step $t$. 
  Of course
  $$m_{S_{t+1}}(I_t([a-1])) \le (a-1)(2n-(a-1))$$
  because the $a-1$ fullest cups must have satisfied the invariant (with $k = a - 1$) on round
  $t$. Moreover, because $\fil_{S_{t+1}}(I_{t+1}(a)) = 0$
  $$m_{S_{t+1}}(I_t([a])) = m_{S_{t+1}}(I_t([a-1])).$$
  Combining the above equations, we get that
  $$m_{S_{t+1}}(A) \le (a-1)(2n-(a-1)).$$
  Furthermore, the fill of all cups in $C$ must be at most $1$ at state $I_t$ to be
  less than the fill of the cup in $A$ that zeroed out. Thus, 
  \begin{align*}
      m_{S_{t+1}}(S_{t+1}([k])) & = m_{S_{t + 1}}(AC)\\ 
                                & \le (a-1)(2n-(a-1))+k-a\\
                                &= a(2n-a) +a -2n+a-1 + k -a\\
                                &= a(2n-a) + (k-n) + (a-n) -1\\
                                &< a(2n-a)
  \end{align*}
  as desired. As $k$ increases from $1$ to $n$, $k(2n-k)$ strictly increases (it is a
  quadratic in $k$ that achieves its maximum value at $k=n$).
  Thus $a(2n-a) \le k(2n-k)$ because $a\le k$.
  Therefore,
  $$m_{S_{t+1}}(S_{t+1}([k])) \le k(2n-k).$$
\end{proof}

\begin{clm}
  If no cups in $A$ zero out in round $t$ and $b=0$, then the invariant
  $\mu_{S_{t+1}}(S_{t+1}([k])) \le 2n-k$ holds.
\end{clm}
\begin{proof}
If $b=0$, then $S_{t+1}([k]) = S_t([k])$. 
During round $t$ the emptier removes $a$ units of fill from the cups in $S_t([k])$,
specifically the cups in $A$. The filler cannot have added more than $k$ fill
to these cups, because it can add at most $1$ fill to any given cup. Also, the
filler cannot have added more than $p_t$ fill to the cups because this is the
total amount of fill that the filler is allowed to add. Hence the filler adds
at most $\min(p_t, k) = a+b=a+0=a$ fill to these cups.
Thus the invariant holds:
$$m_{S_{t+1}}(S_{t+1}([k])) \le m_{S_t}(S_t([k]))+a-a \le k(2n-k).$$
\end{proof}

The remaining case, in which no cups in $A$ zero out and $b > 0$ is the most technically interesting.
\begin{clm}
  If no cups in $A$ zero out on round $t$ and $b > 0$, then the invariant
  $\mu_{S_{t+1}}(S_{t+1}([k])) \le 2n-k$ holds.
\end{clm}
\begin{proof}
Because $b>0$ and $a+b \le k$ we have that $a
< k$, and $c = k-a > 0$. Recall that $S_{t+1}([k]) = AC$, so the mass of the
$k$ fullest cups at $S_{t+1}$ is the mass of $AC$ at $S_t$ plus any water added
to cups in $AC$ by the filler, minus any water removed from cups in $AC$ by the
emptier. The emptier removes exactly $a$ units of water from $AC$.
The filler adds no more than $p_t$ units of water to $AC$ (because the filler
adds at most $p_t$ total units of water per round) and the filler also
adds no more than $k = |AC|$ units of water to $AC$ (because the filler adds
at most $1$ unit of water to each of the $k$ cups in $AC$).
Thus, the filler adds no more than $a+b = \min(p_t, k)$ units of water to $AC$.
Combining these observations we have:
\begin{equation}
m_{S_{t+1}}(S_{t+1}([k])) \le m_{S_t}(AC) + b.
\label{eq:emptiereptiessomestufffillerfillssomestuff}
\end{equation}

The key insight necessary to bound this is to notice that larger values for
$m_{S_t}(A)$ correspond to smaller values for $m_{S_t}(C)$ because of the
invariants; the higher fill in $A$ \defn{pushes down} the fill that $C$ can
have. By capturing the pushing-down relationship combinatorially we will achieve the desired inequality.

We can upper bound $m_{S_t}(C)$ by 
\begin{align*}
m_{S_t}(C) & \le \frac{c}{b+c}m_{S_t}(BC) \\
&= \frac{c}{b+c}(m_{S_t}(ABC) - m_{S_t}(A))
\end{align*}
 because
$\mu_{S_t}(C) \le \mu_{S_t}(B)$ without loss of generality by the
interchangeability of cups.
Thus we have 
\begin{align}
  m_{S_t}(AC) &\le m_{S_t}(A) + \frac{c}{b+c}m_{S_t}(BC)\label{eqn:BCdiscounted}\\
  &= \frac{c}{b+c}m_{S_t}(ABC) + \frac{b}{b+c}m_{S_t}(A)\label{eqn:redistributeA}.
\end{align}

Note that the expression in \eqref{eqn:redistributeA} is monotonically
increasing in both $\mu_{S_t}(ABC)$ and $\mu_{S_t}(A)$. Thus, by numerically
replacing both average fills with their extremal values, $2n-|ABC|$ and
$2n-|A|$. At this point the claim can be verified by straightforward (but quite
messy) algebra (and by combining
\eqref{eq:emptiereptiessomestufffillerfillssomestuff} with
\eqref{eqn:redistributeA}). We instead give a more intuitive argument, in which
we examine the right side of \eqref{eqn:BCdiscounted} combinatorially.

 Consider a new configuration of fills $F$ achieved by starting with state
 $S_t$, and moving water from $BC$ into $A$ until $\mu_{F}(A) = 2n-|A|$.
 \footnote{Note that whether or not $F$ satisfies the invariants is
 irrelevant.} This transformation increases (strictly increases if and only if
 we move a non-zero amount of water) the right side of
 \eqref{eqn:BCdiscounted}. In particular, if mass $\Delta \ge 0$ fill is moved
 from $BC$ to $A$, then the right side of \eqref{eqn:BCdiscounted} increases by
 $\frac{b}{b+c} \Delta \ge 0$. Note that the fact that moving water from $BC$
 into $A$ increases the right side of \eqref{eqn:BCdiscounted} formally
 captures the way the system of invariants being proven forces a tradeoff
 between the fill in $A$ and the fill in $BC$---that is, higher fill in $A$
 pushes down the fill that $BC$ (and consequently $C$) can have.

  Since $\mu_F(A)$ is above $\mu_{F}(ABC)$, the greater than average fill of
  $A$ must be counter-balanced by the lower than average fill of $BC$. In
  particular we must have
  $$(\mu_F(A) - \mu_F(ABC))|A| = (\mu_F(ABC) -\mu_F(BC))|BC|.$$
  Note that 
  \begin{align*}
  & \mu_F(A) -\mu_F(ABC) \\
  &= (2n-|A|) - \mu_F(ABC) \\
  &\ge (2n-|A|) - (2n-|ABC|) \\
  &= |BC|.    
  \end{align*}
  Hence we must have 
  $$\mu_F(ABC) - \mu_F(BC) \ge |A|.$$
  Thus 
  \begin{equation}
      \mu_F(BC) \le \mu_F(ABC) - |A| \le 2n-|ABC| -|A|.
      \label{eq:BCispusheddown}
  \end{equation}
  Combing \eqref{eqn:BCdiscounted} with the fact that the transformation from
  $S_t$ to $F$ only increases the right side of \eqref{eqn:BCdiscounted}, along
  with \eqref{eq:BCispusheddown}, we have the following bound:
  \begin{align}
    m_{S_t}(AC)
  &\le m_{F}(A) + c\mu_{F}(BC) \nonumber \\
  &\le a(2n-a) + c(2n-|ABC|-a) \nonumber \\
  &\le (a+c)(2n-a) - c(a+c+b) \nonumber \\
  &\le (a+c)(2n-a-c) - cb. \label{eq:eqnwithcb}
  \end{align}
  
By \eqref{eq:emptiereptiessomestufffillerfillssomestuff} and \eqref{eq:eqnwithcb}, we have that
\begin{align*}
    m_{S_{t+1}}(S_{t + 1}([k])) & \le m_{S_t}(AC) + b \\
                                & \le (a+c)(2n-a-c) - cb + b \\
                                & = k(2n-k) - cb + b \\
                                & \le k(2n-k),
\end{align*}
where the final inequality uses the fact that $c \ge 1$. This completes the proof of the claim. 
  
\end{proof}

We have shown the invariant holds for arbitrary $k$, so given that the
invariants all hold at state $S_t$ they also must all hold at state $S_{t+1}$.
Thus, by induction we have the invariant for all rounds $t\in\mathbb{N}$.
\end{proof}

\section{Oblivious Filler Lower Bound}\label{sec:oblivious}
In this section we prove that, surprisingly, an oblivious filler can
achieve backlog $n^{1-\varepsilon}$ against an arbitrary
``greedy-like" emptying algorithm.

We say an emptier is \defn{$\Delta$-greedy-like} if, whenever there
are two cups with fills that differ by at least $\Delta$, the emptier
never empties from the less full cup without also emptying from the
more full cup. That is, if on some round $t$, there are cups
$c_1, c_2$ with $\fil_{I_t}(c_1) > \fil_{I_t}(c_2) + \Delta$, then a
$\Delta$-greedy-like emptier doesn't empty from $c_2$ on round $t$
unless it also empties from $c_1$ on round $t$. Note that a perfectly
greedy emptier is $0$-greedy-like. We call an emptier
\defn{greedy-like} if in the cup game on $n$ cups it is
$\Delta$-greedy-like for $\Delta \le \frac{1}{128} \log\log\log
n$. The main result of this section is an oblivious filling strategy
that achieves large backlog against any (possibly randomized)
greedy-like emptier.


As a tool in our analysis we define a new variant of the cup game: In
the $p$-processor \defn{$E$-extra-emptyings} \defn{$S$-skip-emptyings}
negative-fill cup game on $n$ cups, the filler distributes $p$ units
of water amongst the cups, but the filler may empty from $p'$ cups for
some $p' \neq p$. In particular the emptier is allowed to do $E$ extra
emptyings and is also allowed to skip $S$ emptyings over the course of
the game. Note that the emptier still cannot empty from the same cup
twice on a single round, and also that note that a
$\Delta$-greedy-like emptier must take into account extra emptyings
and skip emptyings to determine valid moves. Further, note that the
emptier is allowed to skip extra emptyings, although skipping extra
emptyings looks the same as if the extra-emptyings had simply not been
performed.  Let the \defn{regular} cup game be the $0$-extra-emptyings
$\infty$-skip-emptyings cup game: this is the cup game we usually
consider.
Allowing for some extra emptyings, and bounding the number of skip
emptyings is sometimes necessary when analyzing an algorithm that is a
subroutine of a larger algorithm however, hence it sometimes makes
sense to consider games with different values of $E,S$. Unless
explicitly stated otherwise however we are considering the regular cup
game.

The \defn{fill-range} of a set of cups at a state $S$ is $\max_c
\fil_S(c) - \min_c \fil_S(c)$. We call a cup configuration
\defn{$R$-flat} if the fill-range of the cups less than or equal to
$R$; note that in an $R$-flat cup configuration with average fill
$0$ all cups have fills in $[-R, R]$. 

For a $\Delta$-greedy-like emptier let $R_\Delta = 2(2+\Delta)$;
we now prove a key property of these emptiers: there is an
oblivious filling strategy, which we term \defn{$\flatalg$}, that
attains an $R_\Delta$-flat cup configuration against a
$\Delta$-greedy-like emptier, given cups of a known starting
fill-range.


\begin{lemma}
  \label{lem:flatalg}
  Consider an $R$-flat cup configuration in the $p$-processor
  $E$-extra-emptyings $S$-skip-emptyings negative-fill cup game on $n =
  2p$ cups. There is an oblivious filling strategy
  \defn{$\flatalg$} that achieves an $R_\Delta$-flat
  configuration of cups against a $\Delta$-greedy-like emptier in
  running time $O(R+E+S)$. Furthermore,
  $\flatalg$ guarantees that the cup configuration is $R$-flat on every round.
\end{lemma}
\begin{proof}
  If $R \le R_\Delta$ the algorithm does nothing, since the
  desired fill-range is already achieved; for
  the rest of the proof we consider $R > R_\Delta$.

  The filler's strategy is to distribute fill equally amongst all
  cups at every round, placing $p/n = 1/2$ fill in each cup. 
  Let $\ell_t = \min_{c\in S_t} \fil_{S_t}(c)$, $u_t = \max_{c\in S_t} \fil_{S_t}(c)$. 

  First we show that the fill-range of the cups can only increase
  if the fill-range is very small.
  \begin{clm}
    \label{clm:fillRangeUpIFFfillRangeSmall}
    If the fill-range on round $t+1$ is larger than the
    fill-range on round $t$, then $u_{t+1} - \ell_{t+1} \le R_\Delta.$
  \end{clm}
  \begin{proof}
    First we remark that the fill of any cup changes by at most
    $1/2$ from round to round, and in particular $|u_{t+1}-u_t|
    \le 1/2$, $|\ell_{t+1} - \ell_t|\le 1/2$.
    In order for the fill-range to increase, the emptier must have
    emptied from some cup with fill in $[\ell_t, \ell_t + 1]$ without
    emptying from some cup with fill in $[u_t-1, u_t]$; if the emptier
    had not emptied from every cup with fill in $[\ell_t, \ell_t+1]$
    then we would have $\ell_{t+1} = \ell_t + 1/2$ so the
    fill-range cannot have increased, and similarly if the emptier
    had emptied from every cup with fill in $[u_t-1, u_t]$ then we
    would have $u_{t+1} = u_t - 1/2$ so again the fill-range 
    cannot have increased. Because the emptier is $\Delta$-greedy-like
    emptying from a cup with fill at most $\ell_t+1$ and not
    emptying from a cup with fill at least $u_t-1$ implies that
    $u_t-1$ and $\ell_t+1$ differ by at most $\Delta$.
    Thus, 
    $$u_{t+1} - \ell_{t+1} \le u_t +1/2 - (\ell_t -1/2)  \le \Delta + 3 \le R_\Delta.$$
  \end{proof}

  \cref{clm:fillRangeUpIFFfillRangeSmall} implies that if the
  fill-range is at most $R_\Delta$ it will stay at most
  $R_\Delta$, because fill-range cannot increase to a value
  larger than $R_\Delta$. \cref{clm:fillRangeUpIFFfillRangeSmall}
  also implies that until the fill-range is less than $R_\Delta$
  the fill-range must not increase. However the claim does not
  preclude fill-range from remaining constant for many rounds, or
  decreasing, but only by very small amounts. For this reason we
  actually do not use \cref{clm:fillRangeUpIFFfillRangeSmall} in
  the remainder of the proof; nonetheless, the fact that the
  fill-range cannot increase relative to initial fill-range
  during $\flatalg$ is an important property of $\flatalg$. In
  the rest of the proof we establish that the fill-range indeed
  must eventually be at most $R_\Delta$.

  Let $L_t$ be the set of cups $c$ with $\fil_{S_t}(c) \le \ell_t+2+\Delta$, and let
  $U_t$ be the set of cups $c$ with $\fil_{S_t}(c) \ge u_t-2-\Delta$.
  We now prove a key property of the sets $U_t$ and $L_t$: if a cup is in
  $U_t$ or $L_t$ it is also in $U_{t'}, L_{t'}$ for all $t' > t$. This
  follows immediately from \cref{clm:dontlosestuff}.
  \begin{clm}
    \label{clm:dontlosestuff}
    $U_{t} \subseteq U_{t+1}, L_t \subseteq L_{t+1}.$
  \end{clm}
  \begin{proof}
    Consider a cup $c\in U_t$.

    If $c$ is not emptied from, i.e. $\fil(c)$ has increased by
    $1/2$ from the previous round, then
    clearly $c \in U_{t+1}$, because backlog has increased by at most $1/2$, so
    $\fil(c)$ must still be within $2+\Delta$ of the backlog on round $t+1$. 

    On the other hand, if $c$ is emptied from, i.e. $\fil(c)$ has decreased by
    $1/2$, we consider two cases.\\
    \textbf{Case 1:} If $\fil_{S_t}(c) \ge u_t-\Delta -1$, then
    $\fil_{S_t}(c)$ is at least $1$ above the bottom of the
    interval defining which cups belong to $U_t$. The backlog
    increases by at most $1/2$ and the fill of $c$ decreases by $1/2$, so
    $\fil_{S_{t+1}}(c)$ is at least $1-1/2-1/2 = 0$ above the bottom of the
    interval, i.e. still in the interval. \\
    \textbf{Case 2:} On the other hand, if $\fil_{S_t}(c) <
    u_t-\Delta-1$, then every cup with fill in $[u_t-1, u_t]$
    must have been emptied from because the emptier is
    $\Delta$-greedy-like. Therefore the fullest cup
    on round $t+1$ is the same as the fullest cup on round $t$,
    because every cup with fill in $[u_t-1, u_t]$
    has had its fill decrease by $1/2$, and no cup with fill less than
    $u_t-1$ had its fill increase by more than $1/2$. Hence $u_{t+1}
    = u_t -1/2$. Because both $\fil(c)$ and the backlog
    have decreased by $1/2$, the distance between them is
    still at most $\Delta+2$, hence $c\in U_{t+1}$.\\
    The argument for why $L_t \subseteq L_{t+1}$ is symmetric.
  \end{proof}

  Now we show that under certain conditions $u_t$ decreases and
  $\ell_t$ increases.
  \begin{clm}
    \label{clm:howDoLandUchange}
    On any round $t$ where the emptier empties from at least
    $n/2$ cups, if $|U_t| \le n/2$ then $u_{t+1} = u_t - 1/2$.
    On any round $t$ where the emptier empties from at most $n/2$
    cups, if $|L_t| \le n/2$ then $\ell_{t+1} = \ell_t + 1/2$.
  \end{clm}
  \begin{proof}
    Consider a round $t$ where the emptier empties from at least
    $n/2$ cups. If there are at least $n/2$ cups outside of
    $U_t$, i.e. cups with fills in $[\ell_t, u_t-2-\Delta]$, then
    all cups with fills in $[u_t - 2, u_t]$ must be emptied from;
    if one such cup was not emptied from then by the pigeon-hole
    principle some cup outside of $U_t$ was emptied from, which
    is impossible as the emptier is $\Delta$-greedy-like. This
    clearly implies that $u_{t+1} = u_t - 1/2$: no cup with fill
    less than $u_t-2$ has gained enough fill to become the
    fullest cup, and the fullest cup from the previous round has
    lost $1/2$ unit of fill.

    By a symmetric argument, $\ell_{t+1} = \ell_{t} + 1/2$ on a
    round $t$ where the emptier empties at most $n/2$ cups and
    $|L_t| \le n/2$. 
  \end{proof}

  Now we show that eventually $L_t \cap U_t \neq \varnothing$.
  \begin{clm}
    There is a round $t_0 \le O(R + E + S)$ such that $U_{t}
    \cap L_{t} \neq \varnothing$ for all $t\ge t_0$.
  \end{clm}
  \begin{proof}
  We call a round where the emptier empties from $p' \neq p$ cups
  an \defn{unbalanced round}; we call a round that is not
  unbalanced a \defn{balanced} round. 

  Note that there are clearly at most $E+S$ unbalanced rounds.
  We now associate some unbalanced rounds with balanced rounds;
  in particular we define what it means for a balanced round to
  \defn{cancel} an unbalanced round. 
  Let $B = 2(R + \ceil{(1+1/n)(E+S)})$. For $i = 1,2,\ldots,B$ 
  (iterating in ascending order of $i$), if round $i$
  is unbalanced then we say that the first balanced round $j > i$
  that hasn't already been assigned (earlier in the sequential
  process) to cancel another unbalanced round $i' < i$, if any
  such round $j$ exists, \defn{cancels} round $i$. Note that
  cancellation is a one-to-one relation: each unbalanced round is
  cancelled by at most one balanced round and each balanced round
  cancels at most one unbalanced round. We say a balanced round
  $j$ is \defn{cancelling} if it cancels some unbalanced round $i
  < j$, and \defn{non-cancelling} if it does not cancel any
  unbalanced round.

  We claim that there is some $i \in [2(E+S)+1]$ such that among the
  rounds $[B + 2(E+S) + i-1]$ every unbalanced round has been
  cancelled by a balanced round in $[B + 2(E+S) + i-1]$, and such
  that there are $B$ non-cancelling balanced rounds. Note that there
  are at least $B + (E+S)$ balanced rounds in the first $B +
  2(E+S)$ rounds, and thus there are at least $B$ non-cancelling
  balanced rounds among the first $B+2(E+S)$ rounds, due to there
  being at most $E+S$ total unbalanced rounds. Next note that
  there must be at least $1$ non-cancelling balanced round among
  any set of $2(E+S) + 1$ rounds, because there cannot be more
  than $E+S$ unbalanced rounds, and hence there also cannot be
  more than $E+S$ cancelling rounds. Thus there is some $i \in
  [2(E+S)+1]$ such that round $B + 2(E+S) + i-1$ is balanced and
  non-cancelling. For this $i$, we have that all unbalanced
  rounds in $[B + 2(E+S) + i-1]$ are cancelled, and since $B +
  2(E+S) + i-1 \ge B + 2(E+S)$, we have that there are at least
  $B$ balanced non-cancelling rounds in $[B + 2(E+S) + i-1]$.
  These are the desired properties.

  Let $t_e$ be the first round by which there are $B = 2(R +
  \ceil{(E+S)/n})$ balanced non-cancelling rounds; we have shown
  that $t_e \le O(R+E+S)$. Note that the average fill of the cups
  cannot have decreased by more than $E/n$ from its starting
  value; similarly the average fill of the cups cannot have
  increased by more than $S/n$. Because the cups start $R$-flat,
  $u_t$ cannot have decreased by more than $R + E/n$ or else
  backlog would be less than average fill, and identically
  $\ell_t$ cannot have increased by more than $R + S/n$ or else
  anti-backlog would be larger than average fill. Now, by
  \cref{clm:howDoLandUchange} we have that for some $t \le t_e$,
  $|L_t| > n/2$: if $|L_t|\le n/2$ were always true for $t\le
  t_e$, then on every balanced round $\ell_t$ would have
  increased by $1/2$, and since $\ell_t$ increases by at most
  $1/2$ on unbalanced rounds, this implies that in total $\ell_t$
  would have increased by at least $(1/2)2(R + \ceil{(E+S)/n})$,
  which is impossible. By a symmetric argument it is impossible
  that $|U_t| \le n/2$ for all rounds. 

  Since $|U_{t+1}|\ge |U_t|$ and $|L_{t+1}| \ge |L_t|$ by
  \cref{clm:dontlosestuff}, we have that there is some round $t_0
  \le t_e$ such that for all $t \ge t_0$ we have $|U_t|> n/2$ and
  $|L_t|> n/2$. But then we have $U_t\cap L_t \neq \varnothing$, as desired.
  \end{proof}

  If there exists a cup $c \in L_t\cap U_t$, then 
  $$\fil(c) \in [u_t-2-\Delta, u_t] \cap [\ell_t, \ell_t + 2 +
  \Delta].$$ Hence we have that $$\ell_t+2+\Delta \ge
  u_t-2-\Delta.$$ Rearranging, $$u_t - \ell_t \le 2(2+\Delta) =
  R_\Delta.$$ Thus the cup configuration is $R_\Delta$-flat by
  the end of $O(R+E+S)$ rounds.

\end{proof}


Next we describe a simple oblivious filling strategy, that we
call \defn{$\randalg$}, that will be used as a subroutine in
\cref{lem:obliviousBase}; variants of this strategy
are well-known, and similar versions of it can be found in \cite{
BenderFaKu19, DietzRa91, Kuszmaul20}.
\begin{proposition}
  \label{prop:randalg}
  Consider an $R$-flat cup configuration in the regular single-processor
  $\infty$-extra-emptyings $\infty$-skip-emptyings negative-fill cup
  game on $n$ cups with initial average fill $\mu_0$.
  Let $k \in [n]$ be a parameter. Let $d = \sum_{i=2}^k 1/i$.

  There is an oblivious filling strategy \defn{$\randalg(k)$}
  with running time $k-1$ that achieves fill at least $\mu_0 -R +
  d$ in a known cup $c$ with probability at least $1/k!$ if
  we condition on the emptier not performing extra emptyings.
  $\randalg(k)$ achieves fill at most $\mu_0 + R + d$ in this cup
  (unconditionally).

  Furthermore, when applied against a $\Delta$-greedy-like emptier
  with $R = R_\Delta$, $\randalg(k)$ guarantees that
  the cup configuration is $(R + d)$-flat on every round
  (unconditionally).
\end{proposition}
\begin{proof}
  First we condition on the emptier not using extra emptying
  and show that in this case the filler has probability at least
  $1/(k-1)!$ (which we lower bound by $1/k!$ for sake of
  simplicity) of attaining a cup with fill at least $\mu_0 -R +d$.
  The filler maintains an \defn{active set}, initialized to being
  an arbitrary subset of $k$ of the cups. Every round the filler
  distributes $1$ unit of fill equally among all cups in the
  active set. Next the emptier removes $1$ unit of fill from some
  cup, or skips its emptying. Then the filler removes a random
  cup from the active set (chosen uniformly at random from the
  active set). This continues until a single cup $c$ remains in
  the active set. 

  We now bound the probability that $c$ has never been emptied
  from. Assume that on the $i$-th step of this process, i.e. when
  the size of the active set is $n-i+1$, no cups in the active
  set have ever been emptied from; consider the probability that
  after the filler removes a cup randomly from the active set
  there are still no cups in the active set that the emptier has
  emptied from. If the emptier skips its emptying on this round,
  or empties from a cup not in the active set then it is
  trivially still true that no cups in the active set have been
  emptied from. If the cup that the emptier empties from is in
  the active set then with probability $1/(k-i+1)$ it is evicted
  from the active set, in which case we still have that no
  cup in the active set has ever been emptied from. Hence with
  probability at least $1/(k-1)!$ the final cup in the
  active set, $c$, has never been emptied from. In this case, $c$
  will have gained fill $d=\sum_{i=2}^k 1/i$ as claimed. Because
  $c$ started with fill at least $-R+\mu_0$, $c$ now has fill at
  least $-R+ d+\mu_0$. 

  Now note that regardless of if the emptier uses extra emptyings
  $c$ has fill at most $\mu_0 + R + d$, as $c$ starts with fill
  at most $R$, and $c$ gains at most $1/(k-i+1)$ fill on the
  $i$-th round of this process. 

  Now we analyze this algorithm specifically for a
  $\Delta$-greedy-like emptier. 
  Let $\mathcal{A}_t$ be the event that the
  anti-backlog is smaller in $S_{t+1}$ than in $S_t$, let
  $\mathcal{B}_t$ be the event that some cup with fill equal to
  the backlog in $S_{t+1}$ was emptied from on round $t$. If
  $\mathcal{A}_t$ and $\mathcal{B}_t$ are both true on round $t$,
  then by greediness the cups are quite flat. In particular, 
  let $a$ be a cup with fill equal to the anti-backlog in state
  $S_{t+1}$ that was emptied from on round $t$, and let $b$ be
  a cup with fill equal to the backlog in state $S_{t+1}$ that
  was not emptied from on round $t$. By greediness $\fil_{I_t}(a) + \Delta
  > \fil_{I_t}(b)$. Of course $\fil_{I_t}(b) =
  \fil_{S_{t+1}}(b)$, and for $b$ to have fill equal to the backlog
  on round $t+1$, $b$ must have fill less than $1$ below 
  backlog on round $t$. Of course $\fil_{I_t}(a) \le
  \fil_{S_t}(a) + 1$, and for $a$ to have fill equal to the
  anti-backlog on round $t+1$, $a$ must have fill less than $1$
  above the anti-backlog on round $t$. Thus we have that the
  backlog and anti-backlog differ by at most $\Delta + 3 \le
  R_\Delta$ on round $t$, i.e. the cups are $R_\Delta$-flat.

  Consider a round $t_1$ where the cups are not $R_\Delta$-flat.
  Let $t_0$ be the last round that the cups were $R_\Delta$-flat.
  On all rounds $t \in (t_0, t_1)$ at least one of
  $\mathcal{A}_t$ or $\mathcal{B}_t$ must not hold. On a round
  where $\mathcal{A}_t$ does not hold, anti-backlog does not
  decrease and backlog increases by at most $1/(k-t+1)$, so fill
  range increases by at most $1/(k-t+1).$ On a round where
  $\mathcal{B}_t$ does not hold, anti-backlog decreases by at
  most $1$ and backlog decreases by at least $1-1/(k-t+1)$, as
  all cups with fill equal to the backlog in state $S_{t+1}$ were
  emptied from on round $t$, so fill-range increases by at most
  $1/(k-t+1)$. 
  Hence in total fill-range increases by at most $\sum_{i=2}^k
  1/i$ from $R$, i.e. the cups are $(R+d)$-flat on round $t_1$.

\end{proof}

We now give a general algorithm that specifies a useful
transformation of a filling strategy into a new filling strategy
by repeatedly applying the strategy to subsets of the cups. 
The procedure is similar to the procedure used in the Adaptive
Amplification Lemma, although more complicated.

\begin{definition}
  \label{def:rep}
  {\normalfont
  Let $\alg_0$ be an oblivious filling strategy, that can get
  high fill (for some definition of high) in some cup against
  greedy-like emptiers with some probability. We construct a new
  filling strategy \defn{$\rep_\delta(\alg_0)$} ($\rep$ stands
  for \enquote{repetition}) as follows:

  Say we have some configuration of $n\le N$ cups (recall that
  eventually we aim to get large backlog in $N$ cups).
  Let $n_A = \ceil{\delta n}, n_B = \floor{(1-\delta)n}$. Let
  $M=2^{\polylog(N)}$ be a chosen parameter. 
  Initialize $A$ to $\varnothing$ and $B$ to
  being all of the cups. We call $A$ the \defn{anchor set} and
  $B$ the \defn{non-anchor set}. The filler always places $1$
  unit of fill in each anchor cup on each round. The filling
  strategy consists of $n_A$ \defn{donation-processes}, which are
  procedures that result in a cup being \defn{donated} from $B$
  to $A$ (i.e. removed from $B$ and added to $A$). At the start
  of each donation-processes the filler chooses a value $m_0$
  uniformly at random from $[M]$. We say that the filler
  \defn{applies} a filling strategy $\alg$ to $B$ if the
  filler uses $\alg$ on $B$ while placing $1$ unit of fill
  in each anchor cup. During the donation-process the filler
  applies $\alg_0$ to $B$ $m_0$ times, and flattens $B$ by
  applying $\flatalg$ to $B$ for $\Theta(N^2)$ rounds before each
  application of $\alg_0$. At the end of each donation process
  the filler takes the cup given by the final application of
  $\alg_0$ (i.e. the cup that $\alg_0$ guarantees with some
  probability against a certain class of emptiers to have a
  certain high fill), and donates this cup to $A$. 

  We give pseudocode for this algorithm in \cref{alg:rep}.
 }

\begin{algorithm}
  \caption{$\rep_\delta(\alg_0)$}
  \label{alg:rep}
  \begin{algorithmic}
    \State \textbf{Input:} $\alg_0, \delta, N, M, $ set of $n$ cups
    \State \textbf{Output:} Guarantees on the sets $A, B$ (will vary based on $\alg_0$)
    \State
    \State $n_A \gets \ceil{\delta n}, n_B \gets \floor{(1-\delta) n}$
    \State $A \gets \varnothing, B \gets$ all the cups
    \State Always place $1$ unit of fill into each cup in $A$
    \For{$i \in [n_A]$} \Comment Donation-Processes
    \State $m_0 \gets \text{random}([M])$
      \For{$j \in [m_0]$}
        \State Apply $\flatalg$ to $B$ for $\Theta(N^2)$ rounds
        \State Apply $\alg_0$ to $B$
      \EndFor
      \State Donate the cup given by $\alg_0$ from $B$ to $A$
    \EndFor
  \end{algorithmic}
\end{algorithm}

{\normalfont
We say that the emptier \defn{neglects} the anchor set on a round
if it does not empty from each anchor cup. We say that an
application of $\alg_0$ to $B$ is \defn{non-emptier-wasted} if
the emptier does not neglect the anchor set during any round of
the application of $\alg_0$.
}
\end{definition}

We use the idea of repeating an algorithm in two different
contexts. First in \cref{prop:obliviousBase} we prove a result
analogous to that of \cref{prop:adaptiveBase}. In particular, we
show that we can achieve $\omega(1)$ fill in a known cup (with
good probability) by using $\rep_{\delta}(\randalg(k))$ (for
appropriate choice of $\delta, k$) to get large fill in an
unknown cup, and then and then exploiting the emptier's
greedy-like nature to get (slightly smaller) fill in a known cup.
Second, we use $\rep$ in proving the \defn{Oblivious
Amplification Lemma}, a result analogous to the Adaptive
Amplification Lemma. In particular, we show how to take an
algorithm for achieving some backlog, and then achieve higher
backlog by repeating the algorithm many times. Although these
results have deterministic analogues, their proofs are different
and significantly more complex than the corresponding proofs for
the deterministic results.

In the rest of the section our aim is to achieve backlog
$N^{1-\varepsilon}$ in $N$ cups. We will use this value $N$
within all of the following proofs. We use $N$ to refer to the
true number of cups, and $n$ refer to the size of the current
subproblem that we are considering, which is implicitly part of a
larger cup game. One benefit of making the true number of cups
explicit is that if in a subproblem we ever get mass $N^2$ in a
set of cups, then we are done with the entire construction, as
backlog will always be at least $N$ after this. Note that rather
than using the negative-fill cup game we opt to explicitly make
our results be in terms of the average fill of the cups. Another
thing to note is that we are assuming that $N$ is sufficiently
large for all of the results, as is standard with asymptotic
notation.

Before proving \cref{prop:obliviousBase} we analyze
$\rep_{\delta}(\randalg(k))$ in \cref{lem:obliviousBase}. We
remark that although \cref{prop:obliviousBase} and
\cref{lem:obliviousBase} do implicitly consider a small cup game
that is part of a larger cup game, we do not allow for
extra-emptyings in the small cup game. That is, if there are
extra-emptyings, then we provide no guarantees on the behavior of
the algorithms given in \cref{prop:obliviousBase} and
\cref{lem:obliviousBase}.
\begin{lemma}
  \label{lem:obliviousBase}
  Let $\Delta \le \frac{1}{128}\log\log\log N$, let $h =
  \frac{1}{16}\log\log\log N$, let $k=\ceil{e^{2h+1}}$, let
  $\delta = \frac{1}{2k}$, let $n = \Theta(\log^5 N)$. Consider an
  $R_\Delta$-flat cup configuration in the variable-processor cup
  game on $n$ cups with initial average fill $\mu_0$.

  Against a $\Delta$-greedy-like emptier,
  $\rep_{\delta}(\randalg(k))$ using $M = 2^{\Theta(\log^4 N)}$
  either achieves mass $N^2$ in the cups, or with probability at
  least $1-2^{-\Omega(\log^4 N)}$ makes an unknown cup in $A$
  have fill at least $h+\mu_0$ while also
  guaranteeing that $\mu(B) \ge -h/2 + \mu_0$, where $A,B$ are
  the sets defined in \cref{def:rep}. The running time of
  $\rep_{\delta}(\randalg(k))$ is $2^{O(\log^4 N)}$.
\end{lemma}
\begin{proof}
  We use the definitions given in \cref{def:rep}.

  Note that if the emptier neglects the anchor set $N^2$ times,
  or skips $N^2$ emptyings, then the mass of the cups will be at
  least $N^2$, so the filler is done. For the rest of the proof
  we consider the case where the emptier chooses to neglect
  the anchor set fewer than $N^2$ times, and chooses to skip fewer
  than $N^2$ emptyings.

As in \cref{prop:randalg}, we define $d =
\sum_{i=2}^{k} 1/i$; we remark that, because harmonic numbers
grow like $x\mapsto \ln x$, it is clear that $d=\Theta(h)$. We say that an
application of $\randalg(k)$ to $B$ is \defn{lucky} if it
achieves backlog at least $\mu_S(B) - R_\Delta + d$ where $S$
denotes the state of the cups at the start of the application of
$\randalg(k)$; note that by
\cref{prop:randalg} if we condition on an
application of $\randalg(k)$ where $B$ started $R_\Delta$-flat
being non-emptier-wasted then the application has at least a
$1/k!$ chance of being lucky.

Now we prove several important bounds satisfied by $A$ and $B$.
\begin{clm}
  \label{clm:allflatteningsworkbyM}
  All applications of $\flatalg$ make $B$ be $R_\Delta$-flat and
  $B$ is always $(R_\Delta + d)$-flat.
\end{clm}
\begin{proof}
  Given that the application of $\flatalg$ immediately prior to an application
  of $\randalg(k)$ made $B$ be $R_\Delta$-flat, by
  \cref{prop:randalg} we have that $B$ will
  stay $(R_\Delta + d)$-flat during the application of $\randalg(k)$. 
  Given that the application of $\randalg(k)$ immediately prior to an
  application of $\flatalg$ resulted in $B$ being $(R_\Delta
  + d)$-flat, we have that $B$ remains $(R_\Delta + d)$-flat
  throughout the duration of the application of $\flatalg$ by
  \cref{lem:flatalg}. Given that $B$ is $(R_\Delta +
  d)$-flat before a donation occurs $B$ is clearly still $(R_\Delta +
  d)$-flat after the donation, because the only change to $B$ during
  a donation is that a cup is removed from $B$ which cannot increase
  the fill-range of $B$.
  Note that $B$ started $R_\Delta$-flat at the beginning of the
  first donation-process.
  Note that if an application of $\flatalg$ begins with $B$ being
  $(R_\Delta + d)$-flat, then by considering the flattening to
  happen in the $(|B|/2)$-processor $N^2$-extra-emptyings
  $N^2$-skip-emptyings cup game we ensure that it makes $B$ be
  $R_\Delta$-flat.
  Hence we have by induction that $B$ has always been $(R_\Delta
  + d)$-flat and that all flattening processes have made $B$ be
  $R_\Delta$-flat. 
\end{proof}

Now we aim to show that $\mu(B)$ is never very low, which we need
in order to establish that every non-emptier-wasted lucky
application of $\randalg(k)$ gets a cup with high fill.
Interestingly, in order to lower bound $\mu(B)$ we find it
convenient to first upper bound $\mu(B)$, which by greediness and
flatness of $B$ gives an upper bound on $\mu(A)$ which we then
use to get a lower bound on $\mu(B)$.

\begin{clm}
  \label{clm:muBdoesntgettoobig}
  We have always had
  $$\mu(B) \le \mu(AB) + 2.$$
\end{clm}
\begin{proof}
  There are two ways that $\mu(B)-\mu(A B)$ can increase: \\
  \textbf{Case 1:}
  The emptier could empty from $0$ cups in $B$ while emptying
  from every cup in $A$. \\
  \textbf{Case 2:}
  The filler could evict a cup with fill lower than $\mu(B)$ from
  $B$ at the end of a donation-process. \\

  Note that cases are exhaustive, in particular note that if the
  emptier skips more than $1$ emptying then $\mu(B) - \mu(AB)$
  must decrease because $|B| > |AB|/2$, as opposed to in Case 1
  where $\mu(B) - \mu(AB)$ increases.

  In Case 1, because the emptier is $\Delta$-greedy-like,
  $$\min_{a\in A} \fil(a) > \max_{b\in B} \fil(b) - \Delta.$$
  Thus $\mu(B) \le \mu(A) + \Delta$. We can use this to get an
  upper bound on $\mu(B) - \mu(AB)$. We have, 
  \begin{align*}
    \mu(B) &= \frac{\mu(AB) |AB| - \mu(A) |A|}{|B|}\\
           &\le \frac{\mu(AB) |AB| - (\mu(B) - \Delta) |A|}{|B|}.
  \end{align*}
  Rearranging terms:
  $$\mu(B) \paren{1+\frac{|A|}{|B|}} \le \frac{\mu(AB) |AB| + \Delta |A|}{|B|}.$$
  Now, because $|A| \cdot \Delta \le n_A
  \cdot \Delta < n$ (by choosing $\delta$ very small), we have 
  $$\mu(B) \frac{|AB|}{|B|}\le \frac{\mu(AB) |AB| + n}{|B|}.$$
  Isolating $\mu(B)$ we have 
  $$\mu(B) \le \mu(AB) + 1.$$

  Consider the final round on which $B$ is skipped while $A$ is
  not skipped (or consider the first round if there is no such
  round).

  From this round onward the only increase to $\mu(B) - \mu(A
  B)$ is due to $B$ evicting cups with fill well below $\mu(B)$.
  We can upper bound the increase of $\mu(B) - \mu(AB)$ by the
  increase of $\mu(B)$ as $\mu(AB)$ is strictly increasing.

  The cup that $B$ evicts at the end of a
  donation-process has fill at least $\mu(B) - R_\Delta -
  (k-1)$, as the running time of $\randalg(k)$ is $k-1$, and
  because $B$ starts $R_\Delta$-flat by
  \cref{clm:allflatteningsworkbyM}. Evicting a cup
  with fill $\mu(B) - R_\Delta - (k -1)$ from $B$ changes
  $\mu(B)$ by $(R_\Delta + k - 1) / (|B|-1)$ where $|B|$ is the
  size of $B$ before the cup is evicted from $B$. Even if this
  happens on each of the $n_A$ donation-processes $\mu(B)$ cannot
  rise higher than $n_A (R_\Delta + k-1) / (n-n_A)$ which by
  design in choosing $\delta = 1/(2k)$ is at most $1$.

  Thus $\mu(B) \le \mu(AB) + 2$ is always true.

\end{proof}

Now, the upper bound on $\mu(B)-\mu(AB)$ along with the guarantee
that $B$ is flat allows us to bound the highest that a cup in $A$
could rise by greediness, which in turn upper bounds $\mu(A)$
which in turn lower bounds $\mu(B)$. 
\begin{clm}
  \label{clm:muBgreaterthanminushover2}
  We always have
  $$\mu(B) \ge -h/2 + \mu_0.$$
\end{clm}
\begin{proof}
  By \cref{clm:muBdoesntgettoobig} and \cref{clm:allflatteningsworkbyM} 
  we have that no cup in $B$ ever has fill greater than
  $u_B = \mu(A B) + 2 + R_\Delta + d$. 
  Let $u_A = u_B + \Delta + 1$. We claim that the backlog in $A$
  never exceeds $u_A$. Note that $\mu(AB), u_A, u_B$ are
  implicitly functions of the round; $\mu(AB)$ can increase from
  $\mu_0$ if the emptier skips emptyings.

  Consider how high the fill of a cup $c \in A$ could be.
  If $c$ came from $B$ then when it is donated
  to $A$ its fill is at most $u_B$; otherwise, $c$
  started with fill at most $R_\Delta$. Both of these expressions
  are less than $u_A - 1$. Now consider how
  much the fill of $c$ could increase while being in $A$. Because
  the emptier is $\Delta$-greedy-like, if a cup $c\in A$ has fill
  more than $\Delta$ higher than the backlog in $B$ then $c$ must
  be emptied from, so any cup with fill at least $u_B + \Delta =
  u_A - 1$ must be emptied from, and hence $u_A$ upper bounds the
  backlog in $A$. 

  Of course an upper bound on backlog in $A$ also serves as
  an upper bound on the average fill of $A$ as well, i.e.
  $\mu(A) \le u_A$.  Now we have
  \begin{align*}
    \mu(B) &= -\frac{|A|}{|B|} \mu(A) + \frac{|A B|}{|B|}\mu(A B) \\
           &\ge -(\mu(AB) + 3+R_\Delta+d+\Delta) \frac{|A|}{|B|} + \frac{|AB|}{|B|}\mu(AB)\\
           &= -(3+R_\Delta+d + \Delta) \frac{|A|}{|B|} + \mu(AB)\\
           &\ge -h/2 + \mu(AB)
  \end{align*}
  where the final inequality follows because $\mu(AB) \ge 0$, and
  because $|A|/|B|$ is sufficiently small by our choice of $\delta = 1/(2k)$.
  Of course $\mu(AB) \ge \mu_0$ so we have
  $$\mu(B) \ge -h/2 + \mu_0.$$

\end{proof}

Now we show that at least a constant fraction of the
donation-processes succeed with exponentially good probability.
\begin{clm}
  \label{clm:baseChernoffBound}
  By choosing $M = 2^{\Theta(\log^4 N)}$ the filler can guarantee that with
  probability at least $1-2^{-\Omega(\log^4 N)}$, the filler achieves
  fill $h+\mu_0$ in some cup in $A$. 
\end{clm}
\begin{proof}
Now we bound the probability that at least one cup in $A$ has
fill at least $\mu_0 + h$. If the emptier does not
\defn{interfere} with an application of $\randalg(k)$, i.e. the
emptier does not ever neglect the anchor set and put more resources
into $B$ than the filler does during the application, then the
application of $\randalg(k)$ achieves a cup with fill at least
$\mu(B) - R_\Delta + d$. By our lower bound
on $\mu(B)$ from \cref{clm:muBgreaterthanminushover2}, namely
$\mu(B) \ge -h/2 + \mu_0$, and because $d\ge 2h$, and $R_\Delta
\le h/2$ due to $\Delta \le h/8$, we have that a successful
donation process donates a cup with fill at least $\mu_0 + h$ to
$A$. If on each of the $n_A$ donation processes the emptier does
not interfere with the final application of $\randalg(k)$, then 
the probability that at least one donation process successfully
donates a cup with fill at least $\mu_0 + h$ to $A$ is at least 
$$1-(1-1/k!)^{n_A} \ge 1-e^{- n\delta/k!}.$$
Now we aim to show that $e^{-n \delta/k!} \le 2^{-\Omega(\log^4 N)}$.
This is equivalent to showing $n\delta/k! \ge \Omega(\log^4 N)$.
Using $k^k$ as an upper bound for $k!$, and recalling that
$n=\Theta(log^5 N)$, we have that it suffices to show that 
$$k^k \le n\delta / \Omega(\log^4 N) \le \delta O(\log N).$$
Recalling that $\delta = 1/(2k)$ this is equivalent to showing
that
$$k^{k+1} \le O(\log N).$$
Now recalling that 
$$k = \Theta(e^{2h}) = \Theta(e^{\frac{1}{8}\log\log\log N}) =
\Theta((\log\log N)^{1/4}),$$ 
we get 
$$k^{k+1} \le 2^{O((\log\log\log N)(\log\log N)^{1/4})} \le
O(2^{\log\log N}) = O(\log N).$$
Hence we have our desired bound on the probability of getting the
desired backlog in a cup, conditional on the emptier not
interfering with the final application $\randalg(k)$ on each
donation process.

  However, the emptier is allowed to interfere, and if the
  emptier interferes with an application of $\randalg(k)$ then
  the application might not get large fill. In fact, the emptier
  can even interfere conditional on the filler's progress during
  $\randalg(k)$; in particular it could choose to only interfere
  with applications of $\randalg(k)$ that look like they might
  succeed! However, by applying $\randalg(k)$ a random
  number of times in each donation process, chosen from $[M]$,
  where $M = 2N^2 k! 2^{\log^4 N}$, by a Chernoff Bound there are
  at least $N^2 2^{\log^4 N}$ successful applications of
  $\randalg(k)$ in the donation process with incredibly good
  probability: probability at least $1-2^{-2^{\lg^4 N}}.$
  But since the emptier cannot neglect the anchor set more than
  $N^2$ times, the emptier has at best a $2^{-\log^4 N}$ chance
  of interfering with the final application of $\randalg(k)$. 
  Taking a union bound over all donation processes, we have that
  with probability at least $1-2^{-\Omega(\log^4 N)}$ the emptier
  does not interfere with the final application of $\randalg(k)$
  on any donation processes. We previously found that the
  probability of achieving the desired backlog in $A$ conditional
  on the final application of $\randalg(k)$ in each donation
  process not being interfered with; multiplying these
  probabilities gives that with probability at least
  $1-2^{-\Omega(\log^4 N)}$ we achieve a cup in $A$ with fill
  at least $\mu_0 + h$.

\end{proof}

We now analyze the running time of the filling strategy. There
are $n_A$ donation-processes. Each donation-process consists of
$O(M)$ applications of $\randalg(k)$, which each take time
$O(1)$, and $O(M)$ applications of $\flatalg$, which each take
$\Theta(N^2)$ time. Thus overall the algorithm takes time $$n_A
\cdot O(M) (O(1) + O(N^2)) = 2^{O(\log^4 N)},$$ as desired.
  
\end{proof}

Now, using \cref{lem:obliviousBase} we show in
\cref{prop:obliviousBase} that an oblivious filler can achieve
fill $\omega(1)$ in a known cup. 
\begin{proposition}
  \label{prop:obliviousBase}
  Let $H = \frac{1}{128}\log\log\log N$, let $\Delta \le
  \frac{1}{128}\log\log\log N$, let $n = \Theta(\log^5 N)$. 
  Consider an $R_\Delta$-flat cup configuration in the
  variable-processor cup game on $n$ cups with average fill $\mu_0$.
  There is an oblivious filling strategy that either
  achieves mass $N^2$ among the cups, or achieves fill at least $\mu_0 + H$
  in a chosen cup in running time $2^{O(\log^4 N)}$ against a
  $\Delta$-greedy-like emptier with probability at least
  $1-2^{-\Omega(\log^4 N)}.$
\end{proposition}
\begin{proof}
  The filler starts by using $\rep_\delta(\randalg(k))$ with
  parameter settings as in \cref{lem:obliviousBase};
  note that $h$ from \cref{lem:obliviousBase} is $8H$.

  If this results in mass $N^2$ among the cups then the filler is
  already done. Otherwise, with probability at least
  $1-2^{-\Omega(\log^4 N)}$, there is some cup in $A$ with fill
  $\mu_0 + 8H$. We assume for the rest of the proof that there is
  some cup $c_* \in A$ with $\fil(c_*) \ge \mu_0 + 8H$.

  The filler sets $p=1$, i.e. uses a single processor. Now the
  filler exploits the emptier's greedy-like nature to to get fill
  $H$ in a chosen cup $c_0 \in B$. For $5H$ rounds
  the filler places $1$ unit of fill into $c_0$. Because the
  emptier is $\Delta$-greedy-like it must empty from $c_*$ 
  while $\fil(c_*) > \fil(c_0) + \Delta$. Within $5H$ rounds
  the cups $\fil(c_*)$ cannot decrease below $3H+\mu_0 > H + \Delta + \mu_0$.
  Hence, during these $5H$ rounds, only cups with fills larger
  than $H + \mu_0$ can be emptied from by greediness. 
  The fill of $c_0$ started as at least
  $-4H+\mu_0$ as $\mu(B) \ge -h/2+\mu_0$ from
  \cref{lem:obliviousBase}. After $5H$ rounds
  $c_0$ has fill at least $H+\mu_0$, because the emptier cannot
  have emptied $c_0$ until it attained fill $H+\mu_0$, and if
  $c_0$ is never emptied from then it achieves fill $H+\mu_0$.
  Thus the filler achieves fill $H+\mu_0$
  in $c_0$, a \emph{known} cup, as desired.

  The running time is of course still $2^{O(\log^4 N)}$ by
  \cref{lem:obliviousBase}.
\end{proof}

Next we prove the \defn{Oblivious Amplification Lemma}. We remark
that, like \cref{lem:obliviousBase},
\cref{lem:obliviousAmplification} refers to filling strategies in
the \emph{regular} cup game (no extra empties). In particular, as
in \cref{lem:obliviousBase}, we give no guarantees on the
behavior of the algorithms for the $E$-extra-emptyings cup game
with $E > 0$. It is surprising that we do not need to provide
any guarantees; for $\flatalg$ and $\randalg$ we had to prove
that the cups have bounded fill-range regardless of
extra-emptying. We do not need to establish guarantees on the
algorithms when there is extra emptying because
there cannot be very much extra emptying, and by randomly
choosing which applications of our algorithms are
\enquote{important} we can guarantee that all important
applications of our algorithms succeed (with very high
probability).

\begin{lemma}[Oblivious Amplification Lemma]
  \label{lem:obliviousAmplification} 
  Let $\delta \in (0, 1/2)$ be a constant parameter. Let $\Delta
  \le O(1)$. Consider a cup configuration
  in the variable-processor cup game on $n \le N, n >
  \Omega(1/\delta^2)$ cups with average fill $\mu_0$ that is
  $R_\Delta$-flat. Let $\alg(f)$ be an oblivious filling strategy
  that either achieves mass $N^2$ or, with failure probability at
  most $p$ satisfying $p\ge 2^{-O(\log^4 N)}$, achieves backlog $\mu_0 + f(n)$ on such cups
  in running time $T(n)$ against a $\Delta$-greedy-like emptier.
  Let $M = 2^{\Theta(\log^9 (N))}$.

  Consider a cup configuration in the variable-processor cup game
  on $n \le N, n > \Omega(1/\delta^2)$ cups with average fill
  $\mu_0$ that is $R_\Delta$-flat. There exists an oblivious
  filling strategy $\alg(f')$ that either achieves mass $N^2$ or
  with failure probability at most 
  $$p' \le np + 2^{-\log^8 N}$$
  achieves backlog $f'(n)$ satisfying 
  $$f'(n) \ge (1-\delta)^2 f(\floor{(1-\delta)n}) + f(\ceil{\delta n}) + \mu_0$$ 
  and $f'(n) \ge f(n)$, in running time 
  $$T'(n) \le Mn\cdot T(\floor{(1-\delta)n}) + T(\ceil{\delta n})$$
  against a $\Delta$-greedy-like emptier.
\end{lemma}

\begin{proof}
  We use the definitions and notation given in \cref{def:rep}. 

  Note that the emptier cannot neglect the anchor set more than
  $N^2$ times per donation-process, and the emptier cannot skip
  more than $N^2$ emptyings, without causing the mass of the cups
  to be at least $N^2$; we assume for the rest of the proof that
  the emptier chooses not to do this.

  The filler simply uses $\alg(f)$ on all the cups if 
  $$f(n) \ge (1-\delta)^2 f(n_B) + f(n_A).$$
  In this case our strategy trivially has the desired guarantees. 
  In the rest of the proof we consider the case where we cannot
  simply fall back on $\alg(f)$ to achieve the desired backlog.

  The filler's strategy can be summarized as follows:\\
  \textbf{Step 1:} Make $\mu(A) \ge (1-\delta)^2 f(n_B)$ by
  using $\rep_\delta(\alg(f))$ on all the cups,
  i.e. applying $\alg(f)$ repeatedly to $B$, flattening $B$ before
  each application, and then donating a cup from $B$ to $A$ on
  randomly chosen applications of $\alg(f)$.\\
  \textbf{Step 2:} Flatten $A$ using $\flatalg$, and then use
  $\alg(f)$ on $A$.

  Now we analyze Step 1, and show that by appropriately choosing
  parameters it can be made to succeed. For this proof we need
  all donation-processes to succeed. This necessitates choosing
  $M$ very large. In particular we choose $M = 2^{\Theta(\log^9
  N)}$---recall that $[M]$ is the set from which we randomly choose
  how many times to apply $\alg(f)$ in a donation-process. By
  choosing $M$ this large we cannot hope to guarantee that every
  application of $\alg(f)$ succeeds: there are far too many
  applications. On the other hand, having $M$ so large allows us
  to have a very tight concentration bound on how many
  applications of $\alg(f)$ succeed. Ignoring for a moment the
  fact that the emptier can choose to neglect the anchor set,
  i.e. assuming that no applications of $\alg(f)$ are
  emptier-wasted, the probability that fewer than $M\cdot (1-2p)$
  of $M$ applications of $\alg(f)$ succeed is at most
  $e^{-2Mp^2}$ by a Chernoff bound. However, the emptier is
  allowed to \defn{interfere}, i.e. neglect the anchor set and do
  extra emptying in the non-anchor set. Fortunately the emptier
  can interfere with at most $N^2$ of the applications of
  $\alg(f)$, which is very small compared to $M\cdot (1-2p)$.
  Thus, if we condition on there being at least $M(1-2p)$
  applications that would succeed if the emptier did not
  interfere, there are at least $M(1-2p)-N^2$ applications of
  $\alg(f)$ that succeed. Let $\mathcal{W}$ be the event that the
  donation-process succeeds, i.e. the final application of
  $\alg(f)$ is not emptier-wasted and succeeds, and let
  $\mathcal{D}$ be the event that at least $M(1-2p)$ of the $M$
  applications of $\alg(f)$ would succeed without  interference
  by the emptier. Let $1-q = \Pr[\mathcal{W}]$. Obviously
  $$\Pr[\mathcal{W}] \ge \Pr[\mathcal{W} \land \mathcal{D}] =
  \Pr[\mathcal{D}] \cdot \Pr[\mathcal{W} | \mathcal{D}].$$
  Because the filler chooses which application of $\alg(f)$ is
  the final application uniformly at random from $[M]$ we thus
  have $$1-q \ge (1-e^{-2Mp^2})\paren{\frac{M\cdot
  (1-2p)-N^2}{M}}.$$ Rearranging, and over-estimating (i.e.
  dropping unnecessary terms to simplify the expression, while
  maintaining the truth of the expression), we have $$q \le
  e^{-2Mp^2} + 2p + N^2/M.$$ By assumption $p \ge
  2^{-O(\log^4 N)}$, so as $M \ge 2^{\Omega(\log^9 N)}$, we have
  $Mp^2 \ge 2^{\Omega(\log^9 N)}$. Thus,
  $$q \le 2p + 2^{-2^{\Omega(\log^9 N)}} +
  N^2 2^{-\Omega(\log^{9} N)}.$$ 
  For convenience we loosen this to $$q \le 2p + \frac{2^{-\log^8 N}}{N+1}.$$

  Taking a union bound, we have that with failure probability at
  most $q \cdot n_A$ all donation-process successfully achieve a cup
  with fill at least $\mu_{S_0}(B) + f(n_B)$ where $\mu_{S_0}(B)$
  refers to the average fill of $B$ measured at the start of the
  application of $\alg(f)$; now we assume all donation-processes
  are successful, and demonstrate that this translates into the
  desired average fill in $A$.

  Let \defn{$\skips_t$} denote the number of times that the
  emptier has skipped the anchor set by round $t$. Consider how
  $\mu(B) - \skips/n_B$ changes over the course of the donation
  processes. As noted above, at the end of each donation-process
  $\mu(B)$ decreases due to $B$ donating a cup with fill at least
  $\mu(B) + f(n_B)$. In particular, if $S$ denotes the cup state
  immediately before a cup is donated on the $i$-th
  donation-process, $B_0$ denotes the set $B$ before
  the donation and $B_1$ denotes the set $B$ after the donation,
  then $\mu_{S}(B_1) = \mu_{S}(B_0) - f(n_B) / (n-i)$. Now we claim that
  $t\mapsto \mu_{S_t}(B) - \skips_t/n_B$ is monotonically decreasing. 
  Clearly donation decreases $\mu(B) - \skips/n_B$. 
  If the anchor set is neglected then $\mu(B)$ decreases, causing
  $\mu(B) - \skips/n_B$ to decrease. 
  If a skip occurs, then $\skips/n_B$ increases by more than
  $\mu(B)$ increases, causing $\mu(B)-\skips/n_B$ to decrease. 
  Let $t_*$ be the cup state at the end of all the
  donation-processes. We have that 
  \begin{equation}
    \label{eq:harmonic1}
    \mu_{S_{t_*}}(B) - \frac{\skips_{t_*}}{n_B} \le \mu_0 - \sum_{i=1}^{n_A}\frac{f(n_B)}{n-i}.
  \end{equation}
  By conservation of mass we have 
  $$n_A\cdot \mu_{S_{t_*}}(A) + n_B\cdot \mu_{S_{t_*}}(B) = n\mu_0 + \skips_{t_*}.$$
  Rearranging, 
  \begin{equation}
    \label{eq:lowerboundingAharmonic}
    \mu_{S_{t_*}}(A) = \mu_0 + \frac{n_B}{n_A}\paren{\mu_0 +
    \frac{\skips_{t_*}}{n_B} - \mu_{S_{t_*}}(B)}.
  \end{equation}
  Now we obtain a simpler form of
  Inequality~\eqref{eq:harmonic1}. Let $H_n$ denote the $n$-th
  harmonic
  number. We desire a simpler lower bound for 
  $$\sum_{i=1}^{n_A} \frac{1}{n-i} = H_{n-1}-H_{n_B-1}.$$

  We use the well known fact that 
  \begin{equation}
    \label{eq:wellKnownLogIneq}
    \frac{1}{2(n+1)} < H_n - \ln n - \gamma < \frac{1}{2n}
  \end{equation}
  where $\gamma = \Theta(1)$ denotes the Euler-Mascheroni constant.
  Of course $H_{n-1}-H_{n_B-1} \ge H_n - H_{n_B}.$ Now using
  Inequality~\eqref{eq:wellKnownLogIneq} we have
  \begin{align*}
    H_n - H_{n_B} &> \paren{\ln n + \gamma + \frac{1}{2(n+1)}} - \paren{\ln n_B + \gamma + \frac{1}{2n_B}}\\
                  &> \ln \frac{1}{1-\delta} + \frac{1}{2}\paren{\frac{n_B-n-1}{(n+1)n_B}}\\
                  &> \delta - \Theta\paren{\frac{\delta}{(1-\delta)n}}.
  \end{align*}
  Now using this lower bound on $H_n - H_{n_B}$ in
  Inequality~\eqref{eq:lowerboundingAharmonic} we have:
  \begin{align*}
    \mu_{t_*}(A) &> \mu_0 + \frac{n_B}{n_A}\paren{\delta - \Theta\paren{\frac{\delta}{(1-\delta)n}}}f(n_B)\\
                 &= \mu_0 + \frac{\floor{(1-\delta)n}}{\ceil{\delta n}}\paren{\delta - \Theta\paren{\frac{\delta}{(1-\delta)n}}}f(n_B)\\
                 &> \mu_0 + \paren{\frac{1-\delta}{\delta} - \frac{1}{\delta^2 n}}\paren{\delta - \Theta\paren{\frac{\delta}{(1-\delta)n}}}f(n_B)\\
                 &> \mu_0 + \paren{(1-\delta) - \Theta(1/(\delta n))}f(n_B).
  \end{align*}
  Thus, by choosing $n > \Omega(1/\delta^2)$ we have 
  $$\mu_{t_*}(A) > \mu_0 + (1-\delta)^2 f(n_B).$$

We have shown that in Step 1 the filler achieves average fill
$\mu_0 + (1-\delta)f(n_B)$ in $A$ with failure probability at
most $q\cdot n_A$.
Now the filler flattens $A$ and uses $\alg(f)$ on $A$.
It is clear that this is possible, and succeeds with probability
at least $p$.
This gets a cup with fill 
$$\mu_0 + (1-\delta)^2 f(n_B) + f(n_A)$$
in $A$, as desired.

Taking a union bound over the probabilities of Step 1 and Step 2
succeeding gives that the entire procedure fails with probability
at most 
$$p' \le p + q \cdot n_A \le n p + 2^{-\log^8 N}.$$

The running time of Step 1 is clearly $M\cdot n\cdot
T(\floor{(1-\delta)n})$ and the running time of Step 2 is clearly
$T(\ceil{\delta n})$; summing these yields the desired upper
bound on running time.

\end{proof}

Finally we prove that an oblivious filler can achieve backlog
$N^{1-\varepsilon}$, just like an adaptive filler despite the
oblivious filler's disadvantage. The proof is very similar to the
proof of \cref{thm:adaptivePoly}, although significant care must
be taken in the randomized case to get the desired probabilistic
guarantees. 
\begin{theorem}
  \label{thm:obliviousPoly}
  There is an oblivious filling strategy for the
  variable-processor cup game on $N$ cups that achieves backlog
  at least $\Omega(N^{1-\varepsilon})$ for any constant $\varepsilon
  >0$ in running time $2^{\polylog(N)}$ with probability at least
  $1-2^{-\polylog(N)}$ against a $\Delta$-greedy-like emptier
  with $\Delta \le \frac{1}{128} \log\log\log N$.
\end{theorem}
\begin{proof}
  We aim to achieve backlog $(N/n_b)^{1-\varepsilon}-1$ for 
  $n_b = \Theta(\log^5 N)$ on $N$ cups.
  Let $\delta$ be a constant, chosen as a function of $\varepsilon$.

  We call the algorithm from \cref{prop:obliviousBase}
  $\alg(f_0)$. We remark that in the proof of
  \cref{thm:adaptivePoly} the size of the base case is constant,
  whereas now the size of the base case is $\polylog(N)$; the
  larger base case is necessary, as a constant-size base case
  would destroy our probability of success. The probability of
  success achieved by $\alg(f_0)$ is, by construction, sufficient
  to take a union bound over $2^{\polylog N}$ applications of
  $\alg(f_0)$. Now we construct $\alg(f_{i+1})$ as the amplification
  of $\alg(f_i)$ using \cref{lem:obliviousAmplification} and
  parameter $\delta$. Define a sequence $g_i$ as 
  $$g_i =
  \begin{cases}
    n_b\ceil{16/\delta}, & i=0\\
    \floor{ g_{i-1}/(1-\delta) }, & i\ge 1 
  \end{cases}.$$
  We claim the following regarding our construction:
  \begin{clm}
    \label{clm:fikinductionagain}
    \begin{equation}
      f_i(k) \ge (k/n_b)^{1-\varepsilon} - 1 \text{ for all } k \le g_i. \label{eqn:fikinductionAGAIN}
    \end{equation}
  \end{clm}
  \begin{proof}
  We prove \cref{clm:fikinductionagain} by induction on $i$. \\
  For $i=0$, the base case of our induction,
  \eqref{eqn:fikinductionAGAIN} is true as
  $(k/n_b)^{1-\epsilon} - 1 \le O(1)$ for $k\le g_0$, whereas
  $\alg(f_0)$ achieves backlog $\Omega(\log\log\log N) >
  \omega(1)$.\\
  Now we assume \eqref{eqn:fikinductionAGAIN} for $f_i$, and aim to
  establish \eqref{eqn:fikinductionAGAIN} for $f_{i+1}$. 
  Note that, by design of $g_i$, if $k \le g_{i+1}$ then
  $\floor{k\cdot (1-\delta)} \le g_i$. Consider any $k\in
  [g_{i+1}]$. 
  First we deal with the trivial case where $k \le g_0$. Here, 
  $$f_{i+1}(k) \ge f_i(k) \ge \cdots \ge f_0(k) \ge (k/n_b)^{1-\varepsilon} -1.$$
  Now we consider $k \ge g_0$. 
  Since $f_{i+1}$ is the amplification of $f_i$ we have by \cref{lem:obliviousAmplification} that
  $$f_{i+1}(k) \ge (1-\delta)^2 f_i(\floor{(1-\delta)k}) + f_i(\ceil{\delta k}).$$
  By our inductive hypothesis, which applies as $\ceil{\delta k}\le g_i, \floor{k\cdot (1-\delta)} \le g_i$, we have
  $$f_{i+1}(k) \ge (1-\delta)^2 (\floor{(1-\delta)k/n_b}^{1-\varepsilon}-1) + \ceil{\delta k/n_b}^{1-\varepsilon} - 1. $$
  Dropping the floor and ceiling, incurring a $-1$ for dropping the floor, we have
  $$f_{i+1}(k) \ge (1-\delta)^2 (((1-\delta)k/n_b-1)^{1-\varepsilon}-1) + (\delta k/n_b)^{1-\varepsilon} - 1.$$
  Because $(x-1)^{1-\varepsilon} \ge x^{1-\varepsilon} -1$, as $x\mapsto x^{1-\varepsilon}$ is a sub-linear
  sub-additive function, we have 
  $$f_{i+1}(k) \ge (1-\delta)^2 (((1-\delta)k/n_b)^{1-\varepsilon}-2) + (\delta k/n_b)^{1-\varepsilon}-1.$$
  Moving the $(k/n_b)^{1-\varepsilon}$ to the front we have
  $$ f_{i+1}(k) \ge (k/n_b)^{1-\varepsilon} \cdot\left((1-\delta)^{3-\varepsilon} + \delta^{1-\varepsilon} - \frac{2(1-\delta)^2}{(k/n_b)^{1-\varepsilon}} \right) -1.$$
  Because $(1-\delta)^{3-\varepsilon} \ge 1-(3-\varepsilon)\delta$, a fact called Bernoulli's Identity, we have
  $$f_{i+1}(k) \ge (k/n_b)^{1-\varepsilon} \cdot\left(1-(3-\varepsilon)\delta + \delta^{1-\varepsilon} - \frac{2(1-\delta)^2}{(k/n_b)^{1-\varepsilon}} \right)-1.$$
  Of course $-2(1-\delta)^2 > -2$, so 
  $$f_{i+1}(k) \ge (k/n_b)^{1-\varepsilon} \cdot\left(1-(3-\varepsilon)\delta + \delta^{1-\varepsilon} - 2/(k/n_b)^{1-\varepsilon} \right) -1.$$
  Because $$\frac{-2}{(k/n_b)^{1-\varepsilon}} \ge \frac{-2}{(g_0/n_b)^{1-\varepsilon}} \ge -2(\delta/16)^{1-\varepsilon} \ge -\delta^{1-\varepsilon}/2,$$
  which follows from our choice of $g_0 = \ceil{8/\delta} n_b$ and the restriction
  $\varepsilon<1/2$, we have 
  $$f_{i+1}(k) \ge (k/n_b)^{1-\varepsilon} \cdot\left(1-(3-\varepsilon)\delta + \delta^{1-\varepsilon} - \delta^{1-\varepsilon}/2 \right)-1.$$
  Finally, combining terms we have
  $$f_{i+1}(k) \ge  (k/n_b)^{1-\varepsilon} \cdot\left(1-(3-\varepsilon)\delta + \delta^{1-\varepsilon}/2\right)-1. $$
  Because $\delta^{1-\varepsilon}$ dominates $\delta$ for
  sufficiently small $\delta$, there is a choice of
  $\delta=\Theta(1)$ such that 
  $$1-(3-\varepsilon)\delta + \delta^{1-\varepsilon}/2 \ge 1.$$ 
  Taking $\delta$ to be this small we have,
  $$f_{i+1}(k) \ge (k/n_b)^{1-\varepsilon}-1,$$
  completing the proof. 
  \end{proof}

  The sequence $g_i$ larger the sequence $g_i$ from the proof of
  \cref{thm:adaptivePoly} for the same value of $\delta$: it is
  defined by the same recurrence, except that the first term is
  $n_b$ times larger. Thus we have that $g_{i_*} \ge N$ for some
  $i_* \le O(\log N)$. Hence $\alg(f_{i_*})$ achieves backlog
  $$f_{i_*}(N) \ge (N/n_b)^{1-\varepsilon}-1.$$ Let $\varepsilon'
  = 2\varepsilon$. Of course $n_b \le O(N^\varepsilon)$, so 
  $$(N/n_b)^{1-\varepsilon}-1 \ge \Omega(N^{1-\varepsilon'}).$$

  Let the running time of $f_i(N)$ be $T_i(N)$. From the
  Amplification Lemma we have following recurrence bounding $T_i(N)$:
  \begin{align*}
    T_i(n) &\le 2^{\polylog(N)} \cdot T_{i-1}(\floor{(1-\delta)n}) + T_{i-1}(\ceil{\delta n}) \\
            &\le 2^{\polylog(N)}T_{i-1}(\floor{(1-\delta)n}).
  \end{align*}
  It follows that $\alg(f_{i_*})$, recalling that $i_* \le O(\log N)$, has running time
  $$T_{i_*}(n) \le (2^{\polylog(N)})^{O(\log N)} \le 2^{\polylog(N)}$$
  as desired.

  Now we analyze the probability that the construction fails. 
  Consider the recurrence $a_{i+1} = \alpha a_i + \beta, a_0 =
  \gamma$; the recurrence bounding failure probability is a
  special case of this. Expanding, we see that the recurrence
  solves to $a_k = \Theta(\alpha^{k-1})\beta + \alpha^k \gamma$.
  In our case we have 
  $$\alpha \le N, \beta = 2^{-\log^8 N}, \gamma \le 2^{-\Omega(\log^4 N)}.$$
  Hence the recurrence solves to 
  $$p_{i_*} \le 2^{-\polylog(N)},$$
  as desired.

\end{proof}





\bibliographystyle{plain}
\bibliography{paper}
\end{document}